\documentclass[a4paper, final]{article}

\usepackage{a4wide}
\usepackage[utf8]{inputenc}
\usepackage{amsfonts}
\usepackage{amsmath, amsthm}
\usepackage{amssymb} 
\usepackage{enumitem}
\usepackage{microtype}

\theoremstyle{definition}

\newtheorem{theorem}{Theorem}
\newtheorem{lemma}[theorem]{Lemma}

\newtheorem{corollary}[theorem]{Corollary}

\newtheorem{definition}[theorem]{Definition}

\newtheorem*{lemma*}{Lemma}

\newcommand{\cA}{\mathcal{A}}
\newcommand{\bA}{\mathbf{A}}
\newcommand{\cB}{\mathcal{B}}
\newcommand{\cC}{\mathcal{C}}
\newcommand{\cP}{\mathcal{P}}
\newcommand{\hQ}{\widehat{Q}}
\newcommand{\cR}{\mathcal{R}}
\newcommand{\cS}{\mathcal{S}}
\newcommand{\hS}{\widehat{S}}
\newcommand{\hcS}{\widehat{\cS}}
\newcommand{\cT}{\mathcal{T}}

\newcommand{\fa}{{a}}

\newcommand{\ba}{\bar{a}}
\newcommand{\bq}{\bar{q}}
\newcommand{\br}{\bar{r}}
\newcommand{\bs}{\bar{s}}

\newcommand{\bx}{\bar{x}}

\newcommand{\equals}{=}
\newcommand{\<}{\langle}
\renewcommand{\>}{\rangle}

\newcommand{\Mapsto}{{\mathop{\mapsto}}}

\newcommand{\Real}{\mathbb{R}}
\newcommand{\Rational}{\mathbb{Q}}

\newcommand{\vars}{\text{\upshape{vars}}}
\newcommand{\consts}{\text{\upshape{consts}}}
\newcommand{\bconsts}{\text{\upshape{bconsts}}}
\newcommand{\fconsts}{\text{\upshape{fconsts}}}

\newcommand{\hbeta}{\widehat{\beta}}
\newcommand{\hchi}{\widehat{\chi}}

\newcommand{\semequiv}{\mathrel{|}\joinrel\Relbar\joinrel\mathrel{|}}

\newcommand{\JA}{{\mathcal{J}_\cA}}
\newcommand{\JB}{{\mathcal{J}_\cB}}
\newcommand{\Jgamma}{{\mathcal{J}_\gamma}}

\newcommand{\Fr}{\text{fr}}
\newcommand{\floor}[1]{\lfloor #1 \rfloor}
\newcommand{\ceil}[1]{\lceil #1 \rceil}

\newcommand{\rel}{\mathrel{\triangleleft}}

\newcommand{\BsrSlr}{\text{BSR}(\text{SLR})}
\newcommand{\BsrBd}{\text{BSR}(\text{BD})}

\newcommand{\CC}{\text{\textsc{cc}}}

\newcommand{\inv}{\text{\upshape{inv}}}
\newcommand{\Reach}{\text{\upshape{Reach}}}
\newcommand{\Loc}{\text{\upshape{Loc}}}

\newcommand{\hook}{\hookrightarrow}

\newcommand{\hsimeq}{\mathrel{\widehat{\simeq}}}

\newcommand{\Ndef}{N_{\text{def}}}

\begin{document}
\title{The Bernays--Sch\"onfinkel--Ramsey Fragment\\ with Bounded Difference Constraints\\ over the Reals is Decidable}
\author{
	Marco Voigt\\
	\small\textit{Max Planck Institute for Informatics, Saarland Informatics Campus, Saarbr\"ucken, Germany,}\\
	\small\textit{Saarbr\"ucken Graduate School of Computer Science}
}	
\date{}
\maketitle

\begin{abstract}
	First-order linear real arithmetic enriched with uninterpreted predicate symbols yields an interesting modeling language. However, satisfiability of such formulas is undecidable, even if we restrict the  uninterpreted predicate symbols to arity one. In order to find decidable fragments of this language, it is necessary to restrict the expressiveness of the arithmetic part. One possible path is to confine arithmetic expressions to difference constraints of the form $x - y \rel c$, where $\rel$ ranges over the standard relations $<, \leq, =, \neq, \geq, >$ and $x,y$ are universally quantified. However, it is known that combining difference constraints with uninterpreted predicate symbols yields an undecidable satisfiability problem again. In this paper, it is shown that satisfiability becomes decidable if we in addition bound the ranges of universally quantified variables. As bounded intervals over the reals still comprise infinitely many values, a trivial instantiation procedure is not sufficient to solve the problem.
\end{abstract}


\section{Introduction}\label{section:introduction}

It has been discovered about half a century ago that linear arithmetic with additional uninterpreted predicate symbols has an undecidable satisfiability problem~\cite{Putnam1957}.
Even enriching Presburger arithmetic with only a single uninterpreted predicate symbol of arity one suffices to facilitate encodings of the halting problem for two-counter machines~\cite{Downey1972,VoigtArXiv2017}.
These results do not change substantially when we use the reals as underlying domain instead of the integers.
This means, in order to obtain a decidable subfragment of the combination of linear arithmetic with uninterpreted predicate symbols, the arithmetic part has to be restricted considerably.
In this paper, two subfragments with a decidable satisfiability problem are presented.
Both are based on the Bernays--Sch\"onfinkel--Ramsey fragment (BSR) of first-order logic, which is the $\exists^* \forall^*$ prefix class.
Uninterpreted constant symbols and the distinguished equality predicate are allowed, non-constant function symbols are not.
The arity of uninterpreted predicate symbols is not restricted.
We extend BSR in two ways and call the obtained fragments \emph{BSR modulo simple linear real constraints (\BsrSlr)} and \emph{BSR modulo bounded difference constraints (\BsrBd)}.

	The first clause class---defined in Definition~\ref{definition:BsrSlaSyntax} and treated in detail in Section~\ref{section:DecidabilityBsrSla}---adds constraints of the form $s \rel t$, $x \rel t$, and $x \rel y$ to BSR clauses, where $x$ and $y$ are real-valued variables that are implicitly universally quantified, $s$ and $t$ are linear arithmetic terms that are ground, and $\rel {\in} \{<, \leq, =, \neq, \geq,$ $>\}$.
	We allow Skolem constants in the ground terms $s$ and $t$.
	Since their value is not predetermined, they can be conceived as being existentially quantified.
	The constraints used in this clause fragment are similar to the kind of constraints that appear in the context of the \emph{array property fragment}~\cite{Bradley2006} and extensions thereof (see, e.g.,~\cite{Ge2009,VoigtCADE2017}).
	The main differences are that we use the real domain in this paper instead of the integer domain, and that we allow strict inequalities and disequations between universally quantified variables.
	In the presence of uninterpreted function symbols, strict inequality or disequations can be used to assert that some uninterpreted function $f$ is injective.
	This expressiveness prevents certain instantiation-based approaches to satisfiability checking from being applicable, e.g.\ the methods in~\cite{Bradley2006,VoigtCADE2017}.
	In the context of the array property fragment, this expressiveness even leads to undecidability.

	The \BsrBd\ clause class---presented in Definition~\ref{definition:BsrBdSyntax} and in Section~\ref{section:DecidabilityBsrBd}---adds constraints of the form $x \rel c$, $x \rel y$ and $x - y \rel c$ to BSR clauses, where $x$ and $y$ are real-valued variables, $c$ could be any rational number, and $\rel$ ranges over $\{<, \leq, =, \neq, \geq, >\}$ again.
	We refer to constraints of the form $x - y \rel c$ as \emph{difference constraints}.
	Already in the seventies, Pratt identified difference constraints and boolean combinations thereof as an important tool for the formalization of verification conditions~\cite{Pratt1977}.
	Applications include the verification of timed systems and scheduling problems (see, e.g.,~\cite{Kroening2016}    for references).
	As unrestricted combinations of uninterpreted predicate symbols with difference constraints lead to an undecidable satisfiability problem (once more, two-counter machines can be encoded in a simple way~\cite{VoigtArXiv2015}), we have to further confine the language.
	Every difference constraint $x - y \rel c$ has to be conjoined with four additional constraints $c_x \leq x$, $x \leq d_x$, $c_y \leq y$, $y \leq d_y$, where $c_x, d_x, c_y, d_y$ are rationals.
	This restriction seems to weaken expressiveness severely.
	Indeed, it has to, since we aim for a decidable satisfiability problem.
	Yet, we show in Section~\ref{section:ReachabilityForTimedAutomata} that \BsrBd\ clause sets are expressive enough to formulate the reachability problem for timed automata.
	In~\cite{Niebert2002} an encoding of the reachability problem for timed automata in \emph{difference logic} (boolean combinations of difference constraints \emph{without} uninterpreted predicate symbols) is given, which facilitates deciding bounded reachability, i.e.\ the problem of reaching a given set of states within a bounded number of transitions.
	When using \BsrBd\ as a modeling language, we do not have to fix an upper bound on the number of steps a priori.
	
The main result of the present paper is that satisfiability of finite \BsrSlr\ clause sets and finite \BsrBd\ clause sets is decidable (Theorems~\ref{theorem:DecidabilityOfBsrSla} and~\ref{theorem:DecidabilityOfBsrBd}), respectively.
The proof technique is very similar for the two fragments.
It is partially based on methods from Ramsey theory, which are briefly introduced in Section~\ref{section:RamseyTheory}.
The used approach may turn out to be applicable to other fragments of BSR modulo linear real arithmetic as well.

In order to facilitate smooth reading, long proofs are only sketched in the main text and presented in full in the appendix. 
The present paper is an extended version of~\cite{VoigtFroCoS2017}.


\section{Preliminaries and notation}\label{section:preliminaries}

Hierarchic combinations of first-order logic with background theories build upon sorted logic with equality \cite{Bachmair1994b,Baumgartner2013,Kruglov2012}. We instantiate this framework with the BSR fragment and linear arithmetic over the reals as the \emph{base theory}. The \emph{base sort $\cR$} shall always be interpreted by the reals $\Real$. For simplicity, we restrict our considerations to a single \emph{free sort $\cS$}, which may be freely interpreted as some nonempty domain, as usual. 

We denote by $V_\cR$ a countably infinite set of base-sort variables.
\emph{Linear arithmetic (LA) terms} are build from rational constants $0, 1, \tfrac{1}{2}, -2, -\tfrac{3}{4}$, etc., the operators $+, -$, and the variables from $V_\cR$.
We moreover allow base-sort constant symbols whose values have to be determined by an interpretation (\emph{Skolem constants}). 
They can be conceived as existentially quantified.
As predicates over the reals we allow the standard relations $<, \leq, \equals, \not\equals, \geq, >$. 

In order to hierarchically extend the base theory by the BSR fragment, we introduce the free sort $\cS$, a countably infinite set $V_\cS$ of \emph{free-sort variables}, 
a finite set $\Omega_\cS$ of \emph{free (uninterpreted) constant symbols of sort $\cS$} and a finite set $\Pi$ of \emph{free predicate symbols} equipped with sort information. 
Note that every predicate symbol in $\Pi$ has a finite, nonnegative arity and can be of a mixed sort over the two sorts $\cR$ and $\cS$, e.g.\ $P : \cR \times \cS \times \cR$.
We use the symbol $\approx$ to denote the built-in equality predicate on $\cS$. 
To avoid confusion, we tacitly assume that no constant or predicate symbol is overloaded, i.e.\ they have a unique sort.

\begin{definition}[BSR with simple linear real constraints---\BsrSlr]\label{definition:BsrSlaSyntax}	
	A \emph{\BsrSlr\ clause} has the form $\Lambda \,\|\, \Gamma \to \Delta$, where $\Lambda$, $\Gamma$, $\Delta$ are multisets of atoms satisfying the following conditions.
	(i) Every atom in $\Lambda$ is an LA constraint of the form $s\rel t$ or $x\rel t$ or $x \rel y$ 
                              where $s, t$ are ground (i.e.\ variable-free) LA terms, $x, y \in V_\cR$, and $\rel \,{\in} \{<, \leq,$ $\equals, \not\equals, \geq, >\}$.
	(ii) Every atom in $\Gamma$ and $\Delta$ is either an equation $s\approx s'$ over free-sort variables and constant symbols,
                              or a non-equational atom $P(s_1, \ldots, s_m)$ that is well sorted and where the $s_i$ range over base-sort variables, free-sort variables, and free-sort constant symbols.
\end{definition}

\begin{definition}[BSR with bounded difference constraints---\BsrBd]\label{definition:BsrBdSyntax}	
		A \emph{\BsrBd\ clause} has the form $\Lambda \,\|\, \Gamma \to \Delta$, where the multisets $\Gamma$, $\Delta$ satisfy Condition~(ii) of Definition~\ref{definition:BsrSlaSyntax}, and every atom in $\Lambda$ is an LA constraint of the form $x \rel c$, $x \rel y$, or $x - y \rel c$ where $c$ may be any rational constant (not a Skolem constant), $x, y \in V_\cR$, and $\rel \,{\in} \{<, \leq,$  $\equals, \not\equals, \geq, >\}$.
		Moreover, we require that whenever $\Lambda$ contains a constraint of the form $x - y \rel c$, then $\Lambda$ also contains constraints $c_x \leq x$, $x \leq d_x$, $c_y \leq y$, and $y \leq d_y$ with $c_x, d_x, c_y, d_y \in \Rational$.
\end{definition}

We omit the empty multiset left of ``$\rightarrow$'' and denote it by $\Box$ right of ``$\rightarrow$'' 
(where $\Box$ at the same time stands for \emph{falsity}).
The introduced clause notation separates arithmetic constraints from the free first-order part. 
We use the vertical double bar ``$\|$'' to indicate this syntactically.  
Intuitively, clauses $\Lambda \,\|\, \Gamma \to \Delta$ can be read as 
$\bigl(\bigwedge\Lambda \wedge \bigwedge\Gamma\bigr) \to \bigvee\Delta$, i.e.\ the multisets $\Lambda, \Gamma$ 
stand for conjunctions of atoms and $\Delta$ stands for a disjunction of atoms. 
Requiring the free parts $\Gamma$ and $\Delta$ of clauses to not contain any base-sort terms apart from variables does not limit expressiveness. 
Every base-sort term $t \not\in V_\cR$ in the free part can safely be replaced by a fresh base-sort variable $x_t$ when an atomic constraint $x_t \equals t$ is added to the constraint part of the clause (a process known as \emph{purification} or \emph{abstraction}~\cite{Bachmair1994b,Kruglov2012}).

A \emph{(hierarchic) interpretation} is an algebra $\cA$ which interprets the base sort $\cR$ as $\cR^\cA = \Real$, assigns real values to all occurring base-sort Skolem constants and interprets all LA terms and constraints in the standard way. 
Moreover, $\cA$ comprises a nonempty domain $\cS^\cA$, assigns to each free-sort constant symbol $c$ in $\Omega_\cS$ a domain element $c^\cA \in \cS^\cA$, and interprets every sorted predicate symbol $P\!:\!\xi_1\times\ldots\times\xi_m$ in $\Pi$ by some set $P^\cA\subseteq \xi_1^\cA\times\ldots\times\xi_m^\cA$. 
Summing up, $\cA$ extends the standard model of linear arithmetic and adopts the standard approach to semantics of (sorted) first-order logics when interpreting the free part of clauses.

Given an  interpretation $\cA$ and a sort-respecting \emph{variable assignment} $\beta: V_\cR\cup V_\cS \to \cR^\cA \cup \cS^\cA$, we write $\cA(\beta)(s)$ to mean the \emph{value of the term $s$ under $\cA$ with respect to the variable assignment $\beta$}. 
The variables occurring in clauses are implicitly universally quantified. 
Therefore, given a clause $C$, we call $\cA$ a \emph{(hierarchic) model of $C$}, denoted $\cA\models C$, if and only if $\cA,\beta\models C$ holds for every variable assignment $\beta$.
For clause sets $N$, we write $\cA\models N$ if and only if $\cA \models C$ holds for every clause $C \in N$.
We call a clause $C$ (a clause set $N$) \emph{satisfiable} if and only if there exists a  model $\cA$ of $C$ (of $N$). 
Two clauses $C,D$ (clause sets $N,M$) are \emph{equisatisfiable} if and only if $C$ ($N$) is satisfiable whenever $D$ ($M$) is satisfiable and vice versa.

Given a \BsrSlr\ or \BsrBd\ clause $C$, we use the following notation: the set of all constant symbols occurring in $C$ is denoted by $\consts(C)$. 
The set $\bconsts(C)$ ($\fconsts(C)$) is the restriction of $\consts(C)$ to base-sort (free-sort) constant symbols. 
We denote the set of all variables occurring in a clause $C$ by $\vars(C)$.
The same notation is used for sets of clauses.

\begin{definition}[Normal form of \BsrSlr\ and \BsrBd\ clauses]\label{definition:BSRwithConstrNormalform}
A \BsrSlr\ or \BsrBd\ clause $\Lambda \,\|\, \Gamma \to \Delta$ is in \emph{normal form} if
(1) all non-ground atoms in $\Lambda$ have the form $x \rel c$, $x \rel y$, or $x - y \rel c$ 
	where $c$ is a rational constant or a Skolem constant, and
(2) every variable that occurs in $\Lambda$ also occurs in $\Gamma$ or in $\Delta$.
A \BsrSlr\ or \BsrBd\ clause set $N$ is in \emph{normal form} if all clauses in $N$ are in normal form and pairwise variable disjoint.
Moreover, we assume that $N$ contains at least one free-sort constant symbol.

For \BsrSlr\ clause sets, we pose the following additional requirement.
$N$ can be divided into two parts $\Ndef$ and $N'$ such that
(a) every clause in $\Ndef$ has the form $c \neq t \,\|\, \rightarrow \Box$ where $c$ is a Skolem constant and $t$ is some ground LA term, and
(b) any ground atom $s \rel t$ in any constraint part $\Lambda$ in any clause $\Lambda \,\|\, \Gamma \rightarrow \Delta$ in $N'$ is such that $s$ and $t$ are constants (Skolem or rational, respectively).
\end{definition}

For every \BsrSlr\ clause set $N$ there is an equisatisfiable \BsrSlr\ clause set $N'$ in normal form.
The same holds for \BsrBd\ clause sets.
Requirement~(2) can be established by any procedure for eliminating existentially quantified variables in LA constraints (see, e.g.,~\cite{Fietzke2012}).
Establishing the other requirements is straightforward.

For two sets $R, Q \subseteq \Real$ we write $R < Q$ if $r < q$ holds for all $r \in R$ and $q \in Q$.
Given a real $r$, we denote the \emph{integral part of $r$} by $\floor{r}$, i.e.\ $\floor{r}$ is the largest integer for which $\floor{r} \leq r$.
By $\Fr(r)$ we denote the \emph{fractional part of $r$}, i.e.\ $\Fr(r) := r - \lfloor r \rfloor$. 
Notice that $\Fr(r)$ is always nonnegative, e.g.\ $\Fr(3.71) = 0.71$, whereas $\Fr(-3.71) = 0.29$. 
Given any tuple $\bar{r}$ of reals, we write $\Fr(\bar{r})$ to mean the corresponding tuple of fractional parts, i.e.\ $\Fr\bigl(\<r_1, \ldots, r_\mu\>\bigr) := \bigl\< \Fr(r_1), \ldots, \Fr(r_\mu) \bigr\>$.
We use the notation $\lfloor \bar{r} \rfloor$ in a component-wise fashion as well.

We write $[k]$ to address the set $\{1, \ldots, k\}$ for any positive integer $k > 0$.
Finally, $\cP$ denotes the power set operator, i.e.\ for any set $S$, $\cP(S)$ denotes the set of all subsets of $S$.


\section{Basic tools from Ramsey theory}\label{section:RamseyTheory}
In this section we establish two technical results based on methods usually applied in Ramsey theory.
We shall use these results later on to prove the existence of  models of a particular kind for \BsrSlr\ or \BsrBd\ clause sets that are finite and satisfiable.
These models meet certain uniformity conditions.
In order to construct them, we rely on the existence of certain finite subsets of $\Real$ that are used to construct prototypical tuples of reals.
These finite subsets, in turn, have to behave nicely as well, since tuples that are not distinguishable by \BsrSlr\ or \BsrBd\ constraints are required to have certain uniformity properties.

A tuple $\<r_1, \ldots, r_m\> \in \Real^m$ is called \emph{ascending} if $r_1 < \ldots < r_m$. 
A \emph{coloring} is a mapping $\chi : S \to \cC$ for some arbitrary set $S$ and some finite set $\cC$.
For the most basic result of this section (Lemma~\ref{lemma:RamseyBasic:One}), we consider an arbitrary coloring $\chi$ of $m$-tuples of real numbers and stipulate the existence of a finite subset $Q \subseteq \Real$ of a given cardinality $n$ such that all ascending $m$-tuples of elements from $Q$ are assigned the same color by $\chi$.

\begin{lemma}\label{lemma:RamseyBasic:One}
	Let $n, m > 0$ be positive integers. 
	Let $\chi: \Real^m \to \cC$ be some coloring.
	For every set $R \subseteq \Real$ of \emph{sufficient size} (either infinite or finite with \emph{sufficiently many elements}) there exists a subset $Q \subseteq R$ of cardinality $n$ such that all ascending tuples $\<r_1, \ldots, r_m\> \in Q^m$ are assigned the same color by $\chi$.
\end{lemma}
\begin{proof}[adaptation of the proof of Ramsey's Theorem on page 7 in \cite{Graham1990}]
	For $n < m$ the lemma is trivially satisfied, since in this case $Q^m$ cannot contain ascending tuples. 
	Hence, we assume $n \geq m$.
	In order to avoid technical difficulties when defining the sequence of elements $s_{m-1}, s_{m}, s_{m+1}, \ldots$ below, we assume for the rest of the proof that $R$ is finite but sufficiently large.
	This assumption does not pose a restriction, as we can always consider a sufficiently large subset of $R$.

	We proceed by induction on $m \geq 1$. The base case $m=1$ is easy, since $\chi$ can assign only finitely many colors to elements in $R$ and thus some color must be assigned at least $\bigl\lfloor \tfrac{|R|}{|\cC|} \bigr\rfloor$ times. 
	Hence, if $R$ contains at least $n |\cC|$ elements, we find a uniformly colored subset $Q$ of size $n$.
	Suppose $m>1$. At first, we pick the $m-2$ smallest reals $s_1 < \ldots < s_{m-2}$ from $R$ and set $S_{m-2} := R \setminus \{ s_1, \ldots, s_{m-2}\}$. 
	Thereafter, we simultaneously construct two \emph{sufficiently long but finite} sequences $s_{m-1}, s_m, s_{m+1}, \ldots$ and $S_{m-1}, S_m, S_{m+1}, \ldots$ as follows: \\
	Given $S_i$, we define $s_{i+1}$ to be the smallest real in $S_i$.\\
	Given $S_i$ and the element $s_{i+1}$, we define an equivalence relation $\sim_i$ on the set $S'_i := S_i \setminus \{ s_{i+1} \}$ so that $s \sim_i s'$ holds if and only if for every sequence of indices $j_1, \ldots, j_{m-1}$ with $1 \leq j_1 < \ldots < j_{m-1} \leq i+1$,  we have $\chi(s_{j_1}, \ldots, s_{j_{m-1}}, s) = \chi(s_{j_1}, \ldots, s_{j_{m-1}}, s')$. 
		This equivalence relation partitions $S'_i$ into at most $|\cC|^{{i+1}\choose{m-1}}$ equivalence classes. 
		We choose one such class with largest cardinality to be $S_{i+1}$.
		
	By construction of the sequence $s_1, s_2, s_3, \ldots$, we must have 
		$\chi(s_{j_1}, \ldots, s_{j_{m-1}},$ $s_{k}) = \chi(s_{j_1}, \ldots,$ $s_{j_{m-1}}, s_{k'})$
	for every sequence of indices $j_1 < \ldots < j_{m-1}$ and all indices $k, k' \geq j_{m-1}+1$. 
	Please note that this covers all ascending $m$-tuples in $\{s_1, s_2, s_3, \ldots\}^m$ starting with $s_{j_1}, \ldots, s_{j_m-1}$, i.e.\ they all share the same color.
	We now define a new coloring $\chi' : \{s_1, s_2, s_3, \ldots\}^{m-1} \to \cC$ so that 
		$\chi'(s_{j_1}, \ldots,$ $s_{j_{m-1}}) := \chi(s_{j_1}, \ldots, s_{j_{m-1}}, s_{j_{m-1}+1})$
	for every sequence of indices $j_1 < \ldots < j_{m-1}$ (in case of $j_{m-1}$ being the index of the last element in the sequence $s_1, s_2, s_3, \ldots$, $\chi'(s_{j_1}, \ldots, s_{j_{m-1}})$ shall be an arbitrary color from $\cC$). 
	By induction, there exists a subset $Q \subseteq \{s_1, s_2, s_3, \ldots\}$ of cardinality $n$, such that every ascending $(m-1)$-tuple $\bar{r} \in Q^{m-1}$ is colored the same by $\chi'$. 
	The definition of $\chi'$ entails that now all ascending $m$-tuples $\bar{r}' \in Q^m$ are colored the same by $\chi$. 
	Hence, $Q$ is the sought set.
\end{proof}

Based on Lemma~\ref{lemma:RamseyBasic:One}, one can derive similar results for more structured ways of coloring tuples of reals.
One such result is given in the next lemma.
Its proof can be found in the appendix.


\begin{lemma}\label{lemma:RamseyPermute:Two}
	Let $n, m, p > 0$ be positive integers, let $\kappa \geq 0$ be a nonnegative integer and let $\chi: \Real^m \to \cC$ be an arbitrary coloring.
	Let $R_1, \ldots, R_p$ be \emph{sufficiently large} but finite subsets of $\Real$.
	Let $q_1, \ldots, q_\kappa$ be fixed reals.
	Let $\varrho_1, \ldots, \varrho_L$ be some enumeration of all mappings $\varrho_j : [m] \to [p+\kappa]\times[m]$ for which $\varrho_{j}(i) = \<k,\ell\>$ with $k > p$ entails $\ell = 1$.
	There exist subsets $Q_1 \subseteq R_1, \ldots, Q_p \subseteq R_p$, each of cardinality $n$, such that for all ascending tuples 
		$\bar{r}_1, \bar{r}'_1 \in Q_1^m, \ldots, \bar{r}_p, \bar{r}'_p \in Q_p^m$
	and the reals $r_{\<p+1,1\>} := q_1, \dots, r_{\<p+\kappa,1\>} := q_\kappa$	
	and every index $j$, $1\leq j\leq L$, we have $\chi\bigl(r_{\varrho_j(1)}, \ldots, r_{\varrho_j(m)}\bigr) = \chi\bigl(r'_{\varrho_j(1)}, \ldots, r'_{\varrho_j(m)}\bigr)$.	
\end{lemma}


\section{Decidability of satisfiability for \BsrSlr\ clause sets}\label{section:DecidabilityBsrSla}

For the rest of this section we fix two positive integers $m, m' > 0$ and some finite \BsrSlr\ clause set $N$ in normal form.
For the sake of simplicity, we assume that all uninterpreted predicate symbols $P$ occurring in $N$ have the sort $P : \cS^{m'} \times \cR^{m}$.
This assumption does not limit expressiveness, as the arity of a predicate symbol $P$ can easily be increased in an (un)satisfiability-preserving way by padding the occurring atoms with additional arguments.
For instance, every occurrence of atoms $P(t_1, \ldots, t_m)$ can be replaced with $P(t_1, \ldots, t_m, v, \ldots, v)$ for some fresh variable $v$ that is added sufficiently often as argument.

Given the \BsrSlr\ clause set $N$, every  interpretation $\cA$ induces a partition of $\Real$ into finitely many intervals: the interpretations of all the rational and Skolem constants $c$ occurring in $N$ yield point intervals that are interspersed with and enclosed by open intervals.
\begin{definition}[$\cA$-induced partition of $\Real$]
	Let $\cA$ be an  interpretation and let $r_1, \ldots, r_k$ be all the values in the set $\{c^\cA \mid c \in \bconsts(N)\}$ in ascending order.
	By $\JA$ we denote the following partition of $\Real$: \\
		\centerline{$\JA := \bigl\{ (-\infty, r_1), [r_1, r_1], (r_1, r_2), [r_2, r_2], \ldots, (r_{k-1}, r_k), [r_k, r_k], (r_k, +\infty) \bigr\}$.}
\end{definition}

The idea of the following equivalence is that equivalent tuples are indistinguishable by the constraints that we allow in the \BsrSlr\ clause set $N$.
\begin{definition}[$\JA$-equivalence, $\sim_\JA$]
	Let $\cA$ be an  interpretation and let $k$ be a positive integer.  
	We call two $k$-tuples $\br, \bq \in \Real^{k}$ \emph{$\JA$-equivalent} if \\
	(i) for every $J \in \JA$ and every $i$, $1 \leq i \leq k$, we have $r_i \in J$ if and only if $q_i \in J$ and \\
	(ii) for all $i,j$, $1 \leq i,j \leq k$ we have $r_i < r_j$ if and only if $q_i <  q_j$.
	
	\smallskip\noindent
	The induced equivalence relation on tuples of positive length is denoted by $\sim_\JA$.
\end{definition}
For every positive $k$ the relation $\sim_\JA$ induces only finitely many equivalence classes on the set of all $k$-tuples over the reals.
We intend to show that, if $N$ is satisfiable, then there is some model $\cA$ for $N$ which does not distinguish between different $\JA$-equivalent tuples.
First, we need some notion that reflects how the  interpretation $\cA$ treats a given tuple $\br \in \Real^m$.
This role will be taken by the coloring $\chi_\cA$, which maps $\br$ to a set of expressions of the form $P\ba$, where $P$ is some predicate symbol occurring in $N$ and $\ba$ is an $m'$-tuple of domain elements from $\cS^\cA$.
The presence of $P\ba$ in the set $\chi_\cA(\br)$ indicates that $\cA$ interprets $P$ in such a way that $P^\cA$ contains the pair $\<\ba, \br\>$.
In this sense, $\chi_\cA(\br)$ comprises all the relevant information that $\cA$ contains regarding the tuple $\br$.

\begin{definition}[$\cA$-coloring $\chi_\cA$]\label{definition:AColoring}
	Given an  interpretation $\cA$, let $\hcS := \{ \fa \in \cS^\cA \mid \text{$\fa = c^\cA$ for some $c \in$}$ $\fconsts(N)\}$
	be the set of all domain elements assigned to free-sort constant symbols by $\cA$.
	The \emph{$\cA$-coloring of $\Real^m$} is the mapping \\
		\centerline{$\chi_\cA : \Real^m \to \cP \{P \ba \mid \text{$\ba \in \hcS^{m'}$ and $P$ is an uninterpreted predicate symbol in $N$}\}$}
	defined such that for every $\br \in \Real^m$ we have $P \ba \in \chi_\cA(\br)$ if and only if $\<\ba, \br\> \in P^\cA$.
\end{definition}
Having the coloring $\chi_\cA$ at hand, it is easy to formulate a uniformity property for a given  interpretation $\cA$.
Two tuples $\br, \br' \in \Real^m$ are treated \emph{uniformly} by $\cA$, if the colors $\chi_\cA(\br)$ and $\chi_\cA(\br')$ agree.
Put differently, $\cA$ does not distinguish $\br$ from $\br'$.

\begin{definition}[$\JA$-uniform interpretation]
	An  interpretation $\cA$ is \emph{$\JA$-uni\-form} if $\chi_\cA$ colors each and every $\sim_\JA$-equivalence class uniformly, i.e.\ for all $\sim_\JA$-equivalent tuples $\br, \br'$ we have $\chi_\cA(\br) = \chi_\cA(\br')$.
\end{definition}

We next show that there exists a $\JB$-uniform model $\cB$ of $N$, if $N$ is satisfiable.
Since such a model does not distinguish between $\JB$-equivalent $m$-tuples, and as there are only finitely many equivalence classes induced by $\sim_\JB$, only a finite amount of information is required to describe $\cB$.
This insight will give rise to a decision procedure that nondeterministically guesses how each and every equivalence class shall be treated by the uniform model.

Given some model $\cA$ of $N$, the following lemma assumes the existence of certain finite sets $Q_i$ with a fixed cardinality which are subsets of the open intervals in $\JA$.
All $\JA$-equivalent $m$-tuples that can be constructed from the reals belonging to the $Q_i$ are required to be colored identically by $\chi_\cA$.
The existence of the $Q_i$ is the subject of Lemma~\ref{lemma:IndistuingishableReals}.

\begin{lemma}\label{lemma:ExistenceOfUniformModeslForBsrSla}
	Let $\lambda$ be the maximal number of distinct base-sort variables in any single clause in $N$.
	In case of $\lambda < m$, we set $\lambda := m$. 
	Let $\cA$ be a  model of $N$.
	Let $J_0, \ldots, J_\kappa$ be an enumeration of all open intervals in $\JA$ sorted in ascending order, i.e.\ $J_0 < \ldots < J_\kappa$. 
	Suppose we are given a collection of finite sets $Q_0, \ldots, Q_\kappa$ possessing the following properties:\\
		(i) $Q_i \subseteq J_i$ and $|Q_i| = \lambda$ for every $i$, $0 \leq i \leq \kappa$.\\
		(ii) Let $Q := \bigcup_i Q_i \cup \{c^\cA \mid c \in \bconsts(N) \}$. 
				For all $\JA$-equivalent $m$-tuples $\bq, \bq' \in Q^m$ we have $\chi_\cA(\bq) = \chi_\cA(\bq')$.
				
	\smallskip\noindent		
	Then we can construct a  model $\cB$ of $N$ that is $\JB$-uniform and that interprets the free sort $\cS$ as a finite set.
\end{lemma}
\begin{proof}[Proof sketch]~\\
	\underline{Claim I:} Let $\mu$ be a positive integer with $\mu \leq \lambda$.  
			Every $\sim_\JA$-equivalence class over $\Real^\mu$ contains some representative lying in $Q^\mu$.
			\strut\hfill$\Diamond$
			
	\smallskip
	Let $\hcS$ denote the set $\{ \fa \in \cS^\cA \mid \text{$\fa = c^\cA$ for some $c \in \fconsts(N)$} \}$.
	We construct the  interpretation $\cB$ as follows:
		$\cS^\cB := \hcS$;
		$c^\cB := c^\cA$ for every constant symbol $c$;
		for every uninterpreted predicate symbol $P$ and for all tuples $\ba \in \hcS^{m'}$ and $\br \in \Real^m$ we pick some tuple $\bq \in Q^m$ with $\bq \sim_\JA \br$, and we define $P^\cB$ so that 
			$\<\ba, \br\> \in P^\cB$ if and only if $\<\ba,\bq\> \in P^\cA$.
	By construction, $\cB$ is $\JB$-uniform.
	
	It remains to show $\cB\models N$. 
	Consider any clause $C = \Lambda \;\|\; \Gamma \to \Delta$ in $N$ and let $\beta$ be any variable assignment ranging over $\cS^\cB \cup \Real$. 
	Starting from $\beta$, we derive a special variable assignment $\hbeta_C$ as follows. 
	Let $x_1, \ldots, x_{\ell}$ be all base-sort variables in $C$.
	By Claim~I, there is some tuple $\<q_1, \ldots, q_{\ell}\> \in Q^{\ell}$ such that $\<q_1, \ldots, q_{\ell}\> \sim_\JA \bigl\< \beta(x_1), \ldots, \beta(x_{\ell}) \bigr\>$.
	We set $\hbeta_C(x_i) := q_i$ for every $x_i$.
	For all other base-sort variables, $\hbeta_C$ can be defined arbitrarily.
	For every free-sort variable $u$ we set $\hbeta_C(u) := \beta(u)$.

	As $\cA$ is a  model of $N$, we get $\cA,\hbeta_C \models C$. By case distinction on why $\cA,\hbeta_C \models C$ holds, one can infer $\cB,\beta\models C$.
	Consequently, $\cB\models N$.	
\end{proof}

In order to show that uniform models always exist satisfiable clause sets $N$, we still need to prove the existence of the sets $Q_i$ mentioned in Lemma~\ref{lemma:ExistenceOfUniformModeslForBsrSla}.
We use Lemma~\ref{lemma:RamseyPermute:Two} to show this.

\begin{lemma}\label{lemma:IndistuingishableReals}
	Let $\cA$ be an interpretation.
	Moreover, let $q_1, \ldots, q_\kappa$ be all reals in $\{c^\cA \mid c \in \bconsts(N)\}$ in ascending order and
	let $J_1, \ldots, J_{\kappa+1}$ be all open intervals in $\JA$ in ascending order, i.e.\ $J_1 < \{q_1\} < J_2 < \{q_2\} < \ldots < J_\kappa < \{q_\kappa\} < J_{\kappa+1}$. 
	Let $\lambda$ be a positive integer.
	There is a collection of finite sets $Q_1, \ldots, Q_{\kappa+1}$ such that the following requirements are met.\\
	(i) For every $i$, $1 \leq i \leq \kappa+1$, we have $Q_i \subseteq J_i$ and $|Q_i| = \lambda$.\\
	(ii) Let $Q := \bigcup_i Q_i \cup \{q_1, \ldots, q_\kappa\}$. For all $\JA$-equivalent $m$-tuples $\bs, \bs' \in Q^m$ we have $\chi_\cA(\bs) = \chi_\cA(\bs')$.
\end{lemma}
\begin{proof}[Proof sketch]
	Let the sets $Q_1, \ldots, Q_{\kappa+1}$ be the $Q_1, \ldots, Q_p$ that we obtain by virtue of Lemma \ref{lemma:RamseyPermute:Two} when we set $n := \lambda$, $p := \kappa+1$, $\chi := \chi_\cA$, $R_1 := J_1, \ldots, R_{\kappa+1} := J_{\kappa+1}$.
	Requirement \ref{enum:lemmaIndistuingishableReals:One} is obviously satisfied for $Q_1, \ldots, Q_{\kappa+1}$.
	
	One can show that for every equivalence class $S \in \Real^m/_{\sim_\JA}$ there is some mapping $\varrho : [m] \to [2\kappa+1] \times [m]$ such that
	\begin{enumerate}[label=(\arabic{*}), ref=(\arabic{*})]
		\item\label{enum:lemmaCorrespondenceMainText:One} whenever $\varrho(i) = \<k,\ell\>$ with $k > \kappa + 1$ then $\ell = 1$, and
		\item\label{enum:lemmaCorrespondenceMainText:Two} for all ascending tuples \\
				$\br_1 = \<r_{\<1,1\>}, \ldots, r_{\<1,m\>}\> \in J_1^m; 
					\ldots; 
					\br_{\kappa+1} = \<r_{\<\kappa+1,1\>}, \ldots, r_{\<\kappa+1,m\>}\> \in J_{\kappa+1}^m$; \\
				$\br_{\kappa+2} = \<r_{\<\kappa+2,1\>}\> = \<q_1\>;
					\ldots;
					\br_{2\kappa+1} = \<r_{\<2\kappa+1,1\>}\> = \<q_\kappa\>$\\
			we have $\<r_{\varrho(1)}, \ldots, r_{\varrho(m)}\> \in S$, and
			
		\item\label{enum:lemmaCorrespondenceMainText:Three} for every tuple $\<s_1, \ldots, s_m\> \in S$ there exist ascending tuples $\br_1, \ldots, \br_{2\kappa+1}$ defined as in \ref{enum:lemmaCorrespondenceMainText:Two} such that $\<s_1, \ldots, s_m\> = \<r_{\varrho(1)}, \ldots, r_{\varrho(m)}\>$.
	\end{enumerate}	
	Consider any $\bs, \bs' \in S$. 
	By \ref{enum:lemmaCorrespondenceMainText:Two}, $\bs$ can be written in the form $\<r_{\varrho(1)}, \ldots, r_{\varrho(m)}\>$ for appropriate values $r_{\<k,\ell\>}$ and $\bs'$ can be represented in the form $\<r'_{\varrho(1)}, \ldots, r'_{\varrho(m)}\>$ for appropriate $r'_{\<k,\ell\>}$. 
	Lemma~\ref{lemma:RamseyPermute:Two} entails 
		$\chi_\cA(\bs) = \chi_\cA(\<r_{\varrho(1)}, \ldots, r_{\varrho(m)}\>) = \chi_\cA(\<r'_{\varrho(1)}, \ldots, r'_{\varrho(m)}\>) = \chi_\cA(\bs')$.
\end{proof}

Lemmas~\ref{lemma:ExistenceOfUniformModeslForBsrSla} and~\ref{lemma:IndistuingishableReals} together entail the existence of some $\JA$-uniform model $\cA \models N$ with a finite free-sort domain $\cS^\cA$, if $N$ is satisfiable.
\begin{corollary}\label{corollary:ExistenceOfUniformModeslForBsrSla}
	If $N$ has a  model, then it has a  model $\cA$ that is $\JA$-uniform and that interprets the sort $\cS$ as some finite set.
\end{corollary}

Given any  interpretation $\cA$, the partition $\JA$ of the reals is determined by the rational constants in $N$ and by the values that $\cA$ assigns to the base-sort Skolem constants in $N$.
Let $d_1, \ldots, d_\lambda$ be all the base-sort Skolem constants in $N$.
If we are given some mapping $\gamma : \{d_1, \ldots, d_\lambda\} \to \Real$, then $\gamma$ induces a partition $\Jgamma$, just as $\cA$ induces $\JA$.
We can easily verify whether $N$ has a  model $\cB$ that is \emph{compatible} with $\gamma$ (i.e.\ $\cB$ assigns the same values to $d_1, \ldots, d_\lambda$) and that is $\JB$-uniform.
Due to the uniformity requirement, there is only a finite number of candidate interpretations that have to be checked.

Consequently, in order to show decidability of the satisfiability problem for finite \BsrSlr\ clause sets in normal form, the only question that remains to be answered is whether it is sufficient to consider a finite number of assignments $\gamma$ of real values to the Skolem constants in $N$. 
Recall that since $N$ is in normal form, we can divide $N$ into two disjoint parts $N_\text{def}$ and $N'$ such that all ground LA terms occurring in $N'$ are either (Skolem) constants or rationals.
Moreover, every clause in $\Ndef$ constitutes a definition $c = t$ of some Skolem constant $c$.
As far as the LA constraints occurring in $N'$ are concerned, the most relevant information regarding the interpretation of Skolem constants is their ordering relative to one another and relative to the occurring rationals.
This means, the clauses in $N'$ cannot distinguish two assignments $\gamma, \gamma'$ if \\
(a) for every Skolem constant $d_i$ and every rational $r$ occurring in $N'$ we have
		(a.1) $\gamma(d_i) \leq r$ if and only if $\gamma'(d_i) \leq r$, and
		(a.2) $\gamma(d_i) \geq r$ if and only if $\gamma'(d_i) \geq r$, and \\
(b) for all $d_i, d_j$ we have that
		$\gamma(d_i) \leq \gamma(d_j)$ if and only if $\gamma'(d_i) \leq \gamma'(d_j)$.
		
This observation leads to the following nondeterministic decision procedure for finite \BsrSlr\ clause sets in normal form:
\begin{enumerate}[label=(\arabic{*}), ref=\arabic{*}]
	\item\label{enum:DecisionProcedure:I} Nondeterministically fix a total preorder $\preceq$ (reflexive and transitive) on the set of all base-sort Skolem constants and rational constants occurring in $N'$.
		
		Define a clause set $N_\preceq$ that enforces $\preceq$ for base-sort Skolem constants, i.e.
			$N_\preceq := \bigl\{ c > c' \,\| \rightarrow \Box \bigm| c \preceq c' \text{, either $c$ or $c'$ or both are Skolem constants} \bigr\}$.

	\item\label{enum:DecisionProcedure:II} Check whether there is some mapping $\gamma : \{d_1, \ldots, d_\lambda\} \to \Real$ such that $\gamma$ is a solution for the clauses in $\Ndef \cup N_{\preceq}$. (This step relies on the fact that linear arithmetic over existentially quantified variables is decidable.)
	\item\label{enum:DecisionProcedure:III} If such an assignment $\gamma$ exists, define an  interpretation $\cB$ as follows.
		\begin{enumerate}[label=(\ref{enum:DecisionProcedure:III}.\arabic{*}), ref=(\ref{enum:DecisionProcedure:III}.\arabic{*})]
			\item\label{enum:DecisionProcedure:III.I} Nondeterministically define $\cS^\cB$ to be some subset of $\fconsts(N)$, i.e.\ use a subset of the Herbrand domain with respect to the free sort $\cS$.
			\item\label{enum:DecisionProcedure:III.II} For every $e \in \fconsts(N)$ nondeterministically pick some $\fa \in \cS^\cB$ and set $e^\cB := \fa$.
			\item\label{enum:DecisionProcedure:III.III} Set $d_i^\cB := \gamma(d_i)$ for every $d_i$.
			\item\label{enum:DecisionProcedure:III.IV} For every uninterpreted predicate symbol $P$ occurring in $N$ nondeterministically define the set $P^\cB$ in such a way that $\cB$ is $\JB$-uniform.
		\end{enumerate}
	\item\label{enum:DecisionProcedure:IV} Check whether $\cB$ is a model of $N$.
\end{enumerate}

\begin{theorem}\label{theorem:DecidabilityOfBsrSla}
	Satisfiability of finite \BsrSlr\ clause sets is decidable.
\end{theorem}


\section{Decidability of satisfiability for \BsrBd\ clause sets}\label{section:DecidabilityBsrBd}

Similarly to the previous section, we fix some finite \BsrBd\ clause set $N$ in normal form for the rest of this section,
and we assume that all uninterpreted predicate symbols $P$ occurring in $N$ have the sort $P : \cS^{m'} \times \cR^{m}$.
Moreover, we assume that all base-sort constants in $N$ are integers.
This does not lead to a loss of generality, as we could multiply all rational constants with the least common multiple of their denominators to obtain an equisatisfiable clause set in which all base-sort constants are integers.
We could even allow Skolem constants, if we added clauses stipulating that every such constant is assigned a value that is (a)~an integer and (b) is bounded from above and below by some integer bounds.
For the sake of simplicity, however, we do not consider Skolem constants here.

Our general approach to decidability of the satisfiability problem for finite \BsrBd\ clause sets is very similar to the path taken in the previous section.
Due to the nature of the LA constraints in \BsrBd\ clause sets, the employed equivalence relation characterizing indistinguishable tuples has to be a different one.
In fact, we use one equivalence relation $\hsimeq_\kappa$ on the unbounded space $\Real^m$ and another equivalence relation $\simeq_\kappa$ on the subspace $(-\kappa-1, \kappa+1)^m$ for some positive integer $\kappa$.
Our definition of the relations $\simeq_\kappa$ and $\hsimeq_\kappa$ is inspired by the notion of clock equivalence used in the context of timed automata (see, e.g.,~\cite{Alur1994J}).
\begin{definition}[bounded region equivalence $\simeq_\kappa$]
	Let $\kappa$ be a positive integer.
	We define the equivalence relation $\simeq_\kappa$ on $(-\kappa-1, \kappa+1)^m$ such that we get
	$\<r_1, \ldots, r_m\> \simeq_\kappa \<s_1, \ldots, s_m\>$ if and only if the following conditions are met: \\
		(i) For every $i$ we have  $\lfloor r_i \rfloor = \lfloor s_i \rfloor$, and $\Fr(r_i) = 0$ if and only if $\Fr(s_i) = 0$.\\
		(ii) For all $i,j$ we have $\Fr(r_i) \leq \Fr(r_j)$ if and only if $\Fr(s_i) \leq \Fr(s_j)$.
\end{definition}
The relation $\simeq_\kappa$ induces only a finite number of equivalence classes over $(-\kappa-1, \kappa+1)^m$.
Over $\Real^m$, on the other hand, an analogous equivalence relation $\simeq_\infty$ would lead to infinitely many equivalence classes.
In order to overcome this problem and obtain an equivalence relation over $\Real^m$ that induces only a finite number of equivalence classes, we use the following compromise.

\begin{definition}[unbounded region equivalence $\hsimeq_\kappa$]
	Let $\kappa$ be a positive integer.
	We define the equivalence relation $\hsimeq_\kappa$ on $\Real^m$ in such a way that \\
	$\<r_1, \ldots, r_m\> \hsimeq_\kappa \<s_1, \ldots, s_m\>$ holds if and only if \\
		(i) for every $i$ either $r_i > \kappa$ and $s_i > \kappa$, or $r_i < -\kappa$ and $s_i < -\kappa$, or the following conditions are met:
			(i.i) $\lfloor r_i \rfloor = \lfloor s_i \rfloor$ and 
			(i.ii) $\Fr(r_i) = 0$ if and only if $\Fr(s_i) = 0$,
			and 
		(ii) for all $i,j$\\
			(ii.i) if $r_i, r_j > \kappa$ or $r_i, r_j < -\kappa$, then $r_i \leq r_j$ if and only if $s_i \leq s_j$, \\
			(ii.ii) if $-\kappa \leq r_i, r_j \leq \kappa$, then $\Fr(r_i) \leq \Fr(r_j)$ if and only if $\Fr(s_i) \leq \Fr(s_j)$.
\end{definition}

Obviously, the equivalence relations $\simeq_\kappa$ and $\hsimeq_\kappa$ coincide on the subspace $(-\kappa, \kappa)^m$.
Over $(-\kappa-1, \kappa+1)^m$ the relation $\simeq_\kappa$ constitutes a proper refinement of $\hsimeq_\kappa$.
Figure~\ref{figure:EquivalenceClasses} depicts the equivalence classes induced by $\simeq_\kappa$ and $\hsimeq_\kappa$ in a two-dimensional setting for $\kappa = 1$.
We need both relations in our approach.
\begin{figure}[h]
\centerline{
\begin{tabular}{p{80pt}p{20ex}p{80pt}}
	\begin{picture}(60, 60)
		\put(1.5,1.5){\dashbox(77,77)[5]{}}
		\multiput(20,20)(20,0){3}{\circle*{3}}
		\multiput(20,40)(20,0){3}{\circle*{3}}
		\multiput(20,60)(20,0){3}{\circle*{3}}
		\linethickness{1.5pt}
		\multiput(20,62)(20,0){3}{\line(0,1){16}}
		\multiput(20,42)(20,0){3}{\line(0,1){16}}
		\multiput(20,22)(20,0){3}{\line(0,1){16}}
		\multiput(20,2)(20,0){3}{\line(0,1){16}}
		\multiput(2,60)(20,0){4}{\line(1,0){16}}
		\multiput(2,40)(20,0){4}{\line(1,0){16}}
		\multiput(2,20)(20,0){4}{\line(1,0){16}}
		\multiput(2,62.2)(20,20){1}{\line(1,1){15.8}}
		\multiput(2,62)(20,20){1}{\line(1,1){16}}
		\multiput(2.2,62)(20,20){1}{\line(1,1){15.8}}
		\multiput(2,42.2)(20,20){2}{\line(1,1){15.8}}
		\multiput(2,42)(20,20){2}{\line(1,1){16}}
		\multiput(2.2,42)(20,20){2}{\line(1,1){15.8}}
		\multiput(2,22.2)(20,20){3}{\line(1,1){15.8}}
		\multiput(2,22)(20,20){3}{\line(1,1){16}}
		\multiput(2.2,22)(20,20){3}{\line(1,1){15.8}}
		\multiput(2,2.2)(20,20){4}{\line(1,1){15.8}}
		\multiput(2,2)(20,20){4}{\line(1,1){16}}
		\multiput(2.2,2)(20,20){4}{\line(1,1){15.8}}
		\multiput(22,2.2)(20,20){3}{\line(1,1){15.8}}
		\multiput(22,2)(20,20){3}{\line(1,1){16}}
		\multiput(22.2,2)(20,20){3}{\line(1,1){15.8}}
		\multiput(42,2.2)(20,20){2}{\line(1,1){15.8}}
		\multiput(42,2)(20,20){2}{\line(1,1){16}}
		\multiput(42.2,2)(20,20){2}{\line(1,1){15.8}}
		\multiput(62,2.2)(20,20){1}{\line(1,1){15.8}}
		\multiput(62,2)(20,20){1}{\line(1,1){16}}
		\multiput(62.2,2)(20,20){1}{\line(1,1){15.8}}
	\end{picture}
	&
	\begin{picture}(0,0)
		\put(28.5,9){\mbox{$\<0,0\>$}}
		\put(28,16){\vector(-3,1){65}}
		\put(49,16){\vector(3,1){65}}
	\end{picture}
	&
	\begin{picture}(60, 60)
		\multiput(20,20)(20,0){3}{\circle*{3}}
		\multiput(20,40)(20,0){3}{\circle*{3}}
		\multiput(20,60)(20,0){3}{\circle*{3}}
		\linethickness{1.5pt}
		\multiput(20,62)(20,0){3}{\line(0,1){16}}
		\multiput(20,42)(20,0){3}{\line(0,1){16}}
		\multiput(20,22)(20,0){3}{\line(0,1){16}}
		\multiput(20,2)(20,0){3}{\line(0,1){16}}
		\multiput(2,60)(20,0){4}{\line(1,0){16}}
		\multiput(2,40)(20,0){4}{\line(1,0){16}}
		\multiput(2,20)(20,0){4}{\line(1,0){16}}
		\multiput(22,42.2)(20,20){1}{\line(1,1){15.8}}
		\multiput(22,42)(20,20){1}{\line(1,1){16}}
		\multiput(22.2,42)(20,20){1}{\line(1,1){15.8}}
		\multiput(2,2.2)(20,20){4}{\line(1,1){15.8}}
		\multiput(2,2)(20,20){4}{\line(1,1){16}}
		\multiput(2.2,2)(20,20){4}{\line(1,1){15.8}}
		\multiput(42,22.2)(20,20){1}{\line(1,1){15.8}}
		\multiput(42,22)(20,20){1}{\line(1,1){16}}
		\multiput(42.2,22)(20,20){1}{\line(1,1){15.8}}
	\end{picture}
\end{tabular}	
}
\caption{
	Left: partition of the set $(-2,2)^2$ induced by $\simeq_1$. 
	Right: partition of $\Real^2$ induced by $\hsimeq_1$.
	Every dot, line segment, and white area represents an equivalence class.}
\label{figure:EquivalenceClasses}
\end{figure}
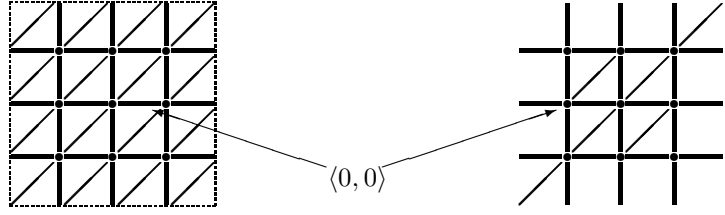
~\\[-7ex]
\begin{definition}[$\simeq_\kappa$-uniform and $\hsimeq_\kappa$-uniform interpretations]
	Let $\kappa$ be a positive integer.
	Consider a  interpretation $\cA$.
	We call $\cA$ \emph{$\simeq_\kappa$-uniform} if its corresponding coloring $\chi_\cA$ (cf.\ Definition~\ref{definition:AColoring}) colors each $\simeq_\kappa$-equivalence class over $(-\kappa-1, \kappa+1)^m$ uniformly, i.e.\ for all tuples $\bq, \bq' \in (-\kappa-1, \kappa+1)^m$ with $\bq \simeq_\kappa \bq'$ we have $\chi_\cA(\bq) = \chi_\cA(\bq')$.
	We call $\cA$ \emph{$\hsimeq_\kappa$-uniform} if $\chi_\cA$ colors each $\simeq_\kappa$-equivalence class over $\Real^m$ uniformly.
\end{definition}

The parameter $\kappa$ will be determined by the base-sort constant in $N$ with the largest absolute value.
If $\kappa$ is defined in this way, one can show that the LA constraints occurring in $N$ cannot distinguish between two $\hsimeq_\kappa$-equivalent $m$-tuples of reals.
This observation is crucial for the proof of Lemma~\ref{lemma:ExistenceOfUniformModeslForBsrBd}.

In order to prove the existence of $\hsimeq_\kappa$-uniform models for satisfiable $N$, we start from some  model $\cA$ of $N$ and rely on the existence of a certain finite set $Q \subseteq [0,1)$ of fractional parts. 
This set $Q$ can be extended to a set $\hQ \subseteq (-\kappa-1, \kappa+1)$ by addition of the fractional parts in $Q$ with integral parts $k$ from the range $-\kappa-1 \leq k \leq \kappa$.
Hence, $\hQ$ contains $2(\kappa+1) \cdot |Q|$ reals.
We assume that all $\simeq_\kappa$-equivalent tuples $\bs, \bs'$ from $\hQ^m$ are treated uniformly by $\cA$.
Put differently, we require $\chi_\cA(\bs) = \chi_\cA(\bs')$.
We choose to formulate this requirement with respect to $\simeq_\kappa$ because of the more regular structure of its equivalence classes, which facilitates a more convenient way of invoking Lemma~\ref{lemma:RamseyBasic:One}.
Due to the fact that $\simeq_\kappa$ constitutes a refinement of $\hsimeq_\kappa$ on the subspace $(-\kappa-1, \kappa+1)^m$, and since for every $\hsimeq_\kappa$-equivalence class $\hS$ over $\Real^m$ there is some $\simeq_\kappa$-equivalence class $S \subseteq (-\kappa-1, \kappa+1)^m$ such that $S \subseteq \hS$, we can use the color $\chi_\cA(\br)$ of representative $m$-tuples $\br$ constructed from $\hQ$ to serve as a blueprint when constructing a $\hsimeq_\kappa$-uniform model $\cB$.


\begin{lemma}\label{lemma:ExistenceOfUniformModeslForBsrBd}
	Let $\lambda$ be the maximal number of distinct base-sort variables in any single clause in $N$; in case of $\lambda < m$, we set $\lambda := m$. 
	Let $\cA$ be a  model of $N$.
	Let $\kappa$ be the maximal absolute value of any rational occurring in $N$; in case this value is zero, we set $\kappa := 1$.
	Suppose we are given a finite set $Q \subseteq [0,1)$ of cardinality $\lambda+1$ such that $0 \in Q$ and for all tuples $\br, \bs \in \hQ^m$,
		$\br \simeq_\kappa \bs$ entails $\chi_\cA(\br) = \chi_\cA(\bs)$,
	where \\
		\centerline{$\hQ := \bigl\{ q + k \bigm| q \in Q \text{ and } k \in \{-\kappa-1, \ldots, 0, \ldots, \kappa \} \bigr\}$.}
	Then we can construct a model $\cB$ of $N$ that is $\hsimeq_\kappa$-uniform and that interprets the free sort $\cS$ as a finite set.
\end{lemma}
\begin{proof}[Proof sketch]
	The construction of $\cB$ from $\cA$ is similar to the construction of uniform models outlined in the proof of Lemma~\ref{lemma:ExistenceOfUniformModeslForBsrSla}.
	
	\smallskip
	\noindent
	\underline{Claim I:} Let $\mu$ be a positive integer with $\mu \leq \lambda$.  
			For every $\hsimeq_\kappa$-equivalence class $S$ over $\Real^\mu$ and every $\br \in S$ there is some $\bq \in S \cap \hQ^\mu$ such that $\br \hsimeq_\kappa \bq$ and for all $i_1, i_2, i_3$ with $r_{i_1} < -\kappa$ and $r_{i_2} > \kappa$ and $-\kappa \leq r_{i_3} \leq \kappa$ we have $\Fr(q_{i_1}) < \Fr(q_{i_2}) < \Fr(q_{i_3})$.
				\strut\hfill$\Diamond$
	\smallskip
	
	Let $\hcS$ denote the set $\{ \fa \in \cS^\cA \mid \text{$\fa = c^\cA$ for some $c \in \fconsts(N)$} \}$.
	We construct the  interpretation $\cB$ as follows:
		$\cS^\cB := \hcS$;
		$c^\cB := c^\cA$ for every constant symbol $c$;
		for every uninterpreted predicate symbol $P$ occurring in $N$ and for all tuples $\ba \in \hcS^{m'}$ and $\br \in \Real^m$ we pick some tuple $\bq \in \hQ^m$ in accordance with Claim~I---i.e.\ $\bq$ satisfies $\br \hsimeq_\kappa \bq$---and define $P^{\cB}$ in such a way that 
			$\<\ba, \br\> \in P^{\cB}$ if and only if $\<\ba,\bq\> \in P^\cA$.
	
	\smallskip
	\noindent
	\underline{Claim II:} The  interpretation $\cB$ is $\hsimeq_\kappa$-uniform.
		\strut\hfill$\Diamond$
	\smallskip
	
	It remains to show $\cB\models N$. 
	We use the same approach as in the proof for Lemma~\ref{lemma:ExistenceOfUniformModeslForBsrSla}, this time based on the equivalence relation $\hsimeq_\kappa$ instead of $\sim_\JA$.
\end{proof}

We employ Lemma~\ref{lemma:RamseyBasic:One} to prove the existence of the set $Q$ used in Lemma~\ref{lemma:ExistenceOfUniformModeslForBsrBd}.

\begin{lemma}\label{lemma:ExistenceOfIndistinguishableFractionalParts}
	Let $\cA$ be an  interpretation and let $\kappa, \lambda$ be positive integers with $\lambda \geq m$.
	There exists a finite set $Q \subseteq [0,1)$ of cardinality $\lambda+1$ such that $0 \in Q$ and for all tuples $\bs, \bs' \in \hQ^m$,
		$\bs \simeq_\kappa \bs'$ entails $\chi_\cA(\bs) = \chi_\cA(\bs')$,
	where \\
		\centerline{$\hQ := \bigl\{ q + k \bigm| q \in Q \text{ and } k \in \{-\kappa-1, \ldots, 0, \ldots, \kappa \} \bigr\}$.}
\end{lemma}
\begin{proof}[Proof sketch]
	One can show that every $\simeq_\kappa$-equivalence class $S$ over $(-\kappa-1, \kappa+1)^m$  can be represented by a pair of mappings $\varrho : [m] \to \{0, 1, \ldots, m\}$ and $\sigma : [m] \to \{-\kappa-1, \ldots, 0, \ldots, \kappa\}$ such that \\
		(i) for any ascending tuple $\<r_0, r_1, \ldots, r_m\> \in [0,1)^{m+1}$ with $r_0 = 0$ we have $\bigl\<r_{\varrho(1)} + \sigma(1), \ldots, r_{\varrho(m)} + \sigma(m)\bigr\> \in S$, and \\			
		(ii) for every tuple $\<s_1, \ldots, s_m\> \in S$ there is an ascending tuple $\<r_0, r_1, \ldots, r_m\> \in [0,1)^{m+1}$ with $r_0 = 0$ such that $\bigl\<s_1, \ldots, s_m\bigr\> = \bigl\<r_{\varrho(1)} + \sigma(1), \ldots, r_{\varrho(m)}+\sigma(m)\bigr\>$.

	Having an enumeration $\<\varrho_1, \sigma_1\>, \ldots, \<\varrho_k, \sigma_k\>$ of pairs of such mappings in which every $\simeq_\kappa$-equivalence class over $(-\kappa-1, \kappa+1)^m$ is represented, we construct a coloring $\hchi : \Real^m \to \bigl( \cP\{P_i \ba \mid \text{$\ba \in \hcS^{m'}$ and $P_i$ occurs in $N$}\} \bigr)^k$ by setting 
		~\\[-4ex]
		\begin{align*}
			\hchi(\br) := \bigl\< \chi_\cA &\bigl( \<r_{\varrho_1(1)} + \sigma_1(1), \ldots, r_{\varrho_1(m)} + \sigma_1(m)\> \bigr),\\
					 &\ldots, \chi_\cA \bigl( \<r_{\varrho_k(1)} + \sigma_k(1), \ldots, r_{\varrho_k(m)} + \sigma_k(m)\> \bigr) \bigr\>
		\end{align*}	
		~\\[-4ex]
	for every tuple $\br = \<r_1, \ldots, r_m\> \in (0,1)^m$, where we define $r_0$ to be $0$.	
	By virtue of Lemma~\ref{lemma:RamseyBasic:One}, there is a set $Q' \subseteq (0,1)$ of cardinality $\lambda$ such that all ascending tuples $\<r_1, \ldots, r_m\> \in {Q'}^m$ are assigned the same color by $\chi$. 
	Then $Q := Q' \cup \{0\}$ is the sought set.
\end{proof}

Lemmas~\ref{lemma:ExistenceOfUniformModeslForBsrBd} and~\ref{lemma:ExistenceOfIndistinguishableFractionalParts} together entail the existence of $\hsimeq_\kappa$-uniform models for finite satisfiable \BsrBd\ clause sets, where $\kappa$ is defined as in Lemma~\ref{lemma:ExistenceOfUniformModeslForBsrBd}.
\begin{corollary}\label{corollary:ExistenceOfUniformModeslForBsrBd}
	Let $\kappa$ be defined as in Lemma~\ref{lemma:ExistenceOfUniformModeslForBsrBd}.
	If $N$ is satisfiable, then it has a  model $\cA$ that is $\hsimeq_\kappa$-uniform and that interprets the sort $\cS$ as some finite set.
\end{corollary}

By virtue of Corollary~\ref{corollary:ExistenceOfUniformModeslForBsrBd}, we can devise a nondeterministic decision procedure for finite \BsrBd\ clause sets $N$. 
We adapt the decision procedure for \BsrSlr\ as follows.
Since base-sort Skolem constants do not occur in $N$, Steps~(\ref{enum:DecisionProcedure:I}),~(\ref{enum:DecisionProcedure:II}), and~\ref{enum:DecisionProcedure:III.III} are skipped.
Moreover, Step~\ref{enum:DecisionProcedure:III.IV} has to be modified slightly.
The interpretations of uninterpreted predicate symbols need to be constructed in such a way that the candidate interpretation $\cB$ is $\hsimeq_\kappa$-uniform for $\kappa := \max\bigl( \{1\} \cup \{ |c| \bigm| c \in \bconsts(N) \} \bigr)$.

\begin{theorem}\label{theorem:DecidabilityOfBsrBd}
	Satisfiability of finite \BsrBd\ clause sets is decidable.
\end{theorem}


\section{Formalizing reachability for timed automata}\label{section:ReachabilityForTimedAutomata}

In this section we show that reachability for timed automata (cf.\ \cite{Alur1994J}) can be formalized using finite \BsrBd\ clause sets.
In what follows, we fix a finite sequence $\bx$ of pairwise distinct \emph{clock variables} that range over the reals.
For convenience, we occasionally treat $\bx$ as a set and use set notation such at $x \in \bx$, $|\bx|$, and $\cP(\bx)$.
A \emph{clock constraint over $\bx$} is a finite conjunction of LA constraints of the form $\texttt{true}$, $x \rel c$, or $x - y \rel c$, where $x,y \in \bx$, $c$ is an integer and $\rel {\in} \{<, \leq, =, \neq, \geq, >\}$.
We denote the \emph{set of all clock constraints over $\bx$} by $\CC(\bx)$.
A \emph{timed automaton} is a tuple
	$\< \Loc, \ell_0, \bx,$ $\<\inv_\ell\>_{\ell\in \Loc}, \cT \>$, 
where
	$\Loc$ is a finite set of locations;
	$\ell_0 \in \Loc$ is the initial location;
	$\<\inv_\ell\>_{\ell \in \Loc}$ is a family of clock constraints from $\CC(\bx)$ where each $\inv_\ell$ describes the \emph{invariant at location $\ell$};
	$\cT \subseteq \Loc \times \CC(\bx) \times \cP(\bx) \times \Loc$ is the location transition relation within the automaton,
		including guards with respect to clocks and the set of clocks that are being reset when the transition is taken.

Although the control flow of a timed automaton is described by finite means, the fact that clocks can assume uncountably many values yields an infinite state space, namely, $\Loc \times [0, \infty)^{|\bx|}$.
Transitions between states fall into two categories:\\
	\begin{tabular}{l@{\hspace{1ex}}l@{\hspace{1ex}}l}
		delay transitions
			&
			$\<\ell, \br\> \hook \<\ell, \br'\>$
			&
			with $\br'= \br+t$ for some $t\geq 0$ and \\&&
			$[\bx' \Mapsto \br'] \models \inv_\ell[\bx']$; and \\[0.5ex]
		location transitions
			&
			$\<\ell, \br\> \hook \<\ell', \br'\>$
			&
			for some $\< \ell, g, Z, \ell' \> \!\in\! \cT$ with $[\bx \Mapsto \br] \models g[\bx]$,  \\&&
			$\br' = \br[Z\mapsto 0]$, and $[\bx' \Mapsto \br'] \models \inv_{\ell'}[\bx']$.			 
	\end{tabular}\\		
	The operation $\br' := \br + t$ is defined by setting $r'_i := r_i + t$ for every $i$,
	and $\br' := \br[Z\mapsto 0]$ means that $r'_i = 0$ for every $x_i \in Z$ and $r'_i = r_i$ for every $x_i \not\in Z$.

In~\cite{Fietzke2012} Fietzke and Weidenbach present an encoding of reachability for a given timed automaton $\bA$ in terms of \emph{first-order logic modulo linear arithmetic}.
\begin{definition}[FOL(LA) encoding of a timed automaton,~\cite{Fietzke2012}]\label{definitionTAclauseSet}
	Given a timed automaton $\bA := \< \Loc, \ell_0, \bx, \<\inv_\ell\>_{\ell\in \Loc}, \cT \>$, the \emph{FOL(LA) encoding of $\bA$} is the following clause set $N_\bA$, where $\Reach$ is a $(1+|\bx|)$-ary predicate symbol:\\[0.5ex]
	\noindent
	\begin{tabular}{l@{\hspace{2ex}}l}
		the initial clause
		&
		$\bigwedge_{x \in \bx} x \equals 0 \;\;\wedge\;\; 
			\inv_{\ell_0}[\bx] \;\bigm\|\; \rightarrow \Reach(\ell_0, \bx)$; \\[0.5ex]
		delay clauses
		&
		$z\geq 0 \;\;\wedge\;\; 
			\bigwedge_{x \in \bx} x' \equals x + z \;\wedge\; 
			\inv_\ell[\bx']$ \\
			&
			\strut\hspace{28ex} $\bigm\|\; \Reach(\ell, \bx) \rightarrow \Reach(\ell, \bx')$ \\
		&
		for every location $\ell\in \Loc$;
	\end{tabular}	
	
	\noindent
	\begin{tabular}{l@{\hspace{2ex}}l}
		transition clauses
		&
		$g[\bx] \;\wedge\;
			\bigwedge_{x\in Z} x' \equals 0 \;\;\wedge
			\bigwedge_{x\in \bx \setminus Z} x' \equals x \;\wedge\;
		 	\inv_{\ell'}[\bx']$ \\
		 	&
		 	\strut\hspace{28ex} $\bigm\|\; \Reach(\ell, \bx) \rightarrow \Reach(\ell', \bx')$ \\
		&
		for every location transition $\<\ell, g, Z, \ell'\> \in \cT$.
	\end{tabular}	
\end{definition}
Corollary~4.3 in~\cite{Fietzke2012} states that for any  model of $N_\bA$, every location $\ell \in \Loc$, and every tuple $\br \in \Real^{|\bx|}$ we have $\cA, [\bx \Mapsto \br] \models \Reach(\ell,\bx)$ if and only if $\bA$ can reach the state $\<\ell, \br\>$ from its initial configuration.

Given any clock constraint $\psi \in \CC(\bx)$ and some location $\ell$, the timed automaton $\bA$ can reach at least one of the states $\<\ell, \br\>$ with $[\bx \Mapsto \br] \models \psi[\bx]$ from its initial configuration if and only if the clause set $N_\bA \cup \bigl\{ \psi[\bx] \,\|\, \Reach(\ell, \bx) \rightarrow \Box \bigr\}$ is unsatisfiable (cf.\ Proposition~4.4 in~\cite{Fietzke2012}).

Next, we argue that the passage of time does not have to be formalized as a synchronous progression of all clocks.
Instead, it is sufficient to require that clocks progress in such a way that their valuations do not drift apart excessively.

\begin{lemma}\label{lemma:DelayClausesWithDifferenceBounds}
	Consider any delay clause \\
		\centerline{$
			 	C :=\quad 
				z\geq 0 \;\;\wedge\;\; 
				\bigwedge_{x \in \bx} x' \equals x + z \;\;\wedge\;\; 
				\inv_\ell[\bx'] 
				\;\bigm\|\; \Reach(\ell, \bx) \rightarrow \Reach(\ell, \bx')
		$}
	that belongs to the FOL(LA) encoding of some timed automaton $\bA := \< \Loc, \ell_0, \bx,$ $\<\inv_\ell\>_{\ell\in \Loc}, \cT \>$.
	Let $\lambda$ be some positive integer.
	Let $M$ be a finite clause set corresponding to the following formula
		\begin{align*}
			~\\[-6ex]
			\varphi :=\;\; \bigwedge_{x_1, x_2 \in \bx}\; &\bigwedge_{-\lambda \leq k \leq \lambda} \bigl(x_1 - x_2 \leq k \;\leftrightarrow\; x'_1 - x'_2 \leq k \bigr) \\[-1.5ex]
				&\hspace{18ex} \wedge\; \bigl(x_1 - x_2 \geq k \;\leftrightarrow\; x'_1 - x'_2 \geq k \bigr) \\[0.25ex]
				&\wedge\; 
					\bigwedge_{x \in \bx} x' \geq x 
					\;\;\wedge\;\; 
					\inv_\ell[\bx'] 
					\;\;\bigm\|\;\; \Reach(\ell, \bx) \rightarrow \Reach(\ell, \bx') ~.
			~\\[-5ex]	
		\end{align*}
	For every $\simeq_\lambda$-uniform interpretation $\cA$ we have 
		$\cA, [\bx \Mapsto \br, \bx' \Mapsto \br'] \models C$
		for all tuples $\br, \br' \in [0, \lambda+1)^{|\bx|}$
	if and only if
		$\cA, [\bx \Mapsto \bq, \bx' \Mapsto \bq'] \models M$
		holds for all tuples $\bq, \bq' \in [0, \lambda+1)^{|\bx|}$.
\end{lemma}

Our approach to decidability of \BsrBd-satisfiability exploits the observation that the allowed constraints cannot distinguish between tuples from one and the same equivalence class with respect to $\hsimeq_\lambda$, which induces only a finite number of such classes.
Decidability of the reachability problem for timed automata can be argued in a similar fashion, using a suitable equivalence relation on clock valuations~\cite{Alur1994J}.
We refer to the induced classes of indistinguishable clock valuations over $\Real^{|x|}$, which are induced by a given timed automaton $\bA = \< \Loc, \ell_0, \bx,$ $\<\inv_\ell\>_{\ell\in \Loc}, \cT \>$, as \emph{TA regions} of $\bA$.

In order to decide reachability for $\bA$, it is sufficient to consider a bounded subspace of $\Real^{|\bx|}$.
More precisely, there exists a computable integer $\lambda$, depending on the number of clocks $|\bx|$ and the constants occurring in clock constraints in $\bA$, such that any valuation $\br$ of $\bA$'s clocks can be projected to some valuation $\br' \in [0, \lambda+1)^{|\bx|}$ that $\bA$ cannot distinguish from $\br$ (see Section~\ref{section:appendix:TAReachability}). 
In the subspace $[0, \lambda+1)^{|\bx|}$, $\bA$'s TA regions coincide with (finite unions of) equivalence classes with respect to $\simeq_\lambda$.
In fact, the quotient $[0, \lambda+1)^{|\bx|}/_{\simeq_\lambda}$ constitutes a refinement of the division of $[0, \lambda+1)^{|\bx|}$ into TA regions.
Since any pair $\<\ell, \br\>$ with $\br \in R$ for some TA region $R$ is reachable if and only if all pairs $\<\ell, \br'\>$ with $\br \in R$ are reachable, any minimal  model $\cA$ of the encoding $N_\bA$ is $\simeq_\lambda$-uniform (where minimality of $\cA$ refers to the minimality of the set $\Reach^\cA$ with respect to set inclusion).
This is why Lemma~\ref{lemma:DelayClausesWithDifferenceBounds} may focus on $\simeq_\lambda$-uniform models.

\begin{theorem}\label{theorem:TAReachabilityInBdrBd}
	The reachability problem for a given timed automaton can be expressed in terms of satisfiability of a finite \BsrBd\ clause set.
\end{theorem}


\newpage

\appendix

\section{Appendix}

\subsection{Details Concerning Section~\ref{section:RamseyTheory}}
\subsubsection*{Proof of Lemma~\ref{lemma:RamseyPermute:Two}}

We start with two auxiliary results.

\begin{lemma}\label{lemma:RamseyBasic:Two}
	Let $n, m, p > 0$ be positive integers and let $\chi: \Real^{m p} \to \cC$ be an arbitrary coloring. 
	Let $R_1, \ldots, R_p$ be \emph{sufficiently large} but finite subsets of $\Real$.
	
	There exist subsets $Q_1 \subseteq R_1, \ldots, Q_p \subseteq R_p$, each of cardinality $n$, such that for all ascending tuples 
		$\bar{r}_1 \in Q_1^m, \ldots, \bar{r}_p \in Q_p^m$ 
	 the colors $\chi(\bar{r}_1, \ldots, \bar{r}_p)$ are the same.	
\end{lemma}
\begin{proof}[adaptation of the proof of Theorem 5 on page 113 in \cite{Graham1990}]~\\
	As in the proof of Lemma \ref{lemma:RamseyBasic:One}, we assume $n \geq m$.
	We proceed by induction on $p \geq 1$.
	
	The case $p=1$ is covered by Lemma \ref{lemma:RamseyBasic:One}.

	Suppose $p > 1$. We define an equivalence relation $\sim_p$ on the set $R_p^m$ so that $\bar{s} \sim_p \bar{s}'$  holds if and only if for all ascending tuples
		$\bar{r}_1 \in R_1^m, \ldots, \bar{r}_{p-1} \in R_{p-1}^m$
	 the colors 
	 	$\chi\bigl(\bar{r}_1, \ldots, \bar{r}_{p-1}, \bar{s}\bigr)$
	 and 
	 	$\chi\bigl(\bar{r}_1, \ldots, \bar{r}_{p-1}, \bar{s}'\bigr)$
	are identical.
	This equivalence relation partitions $R_p^m$ into at most $|\cC|^{{{|R_1|}\choose{m}}\cdot \ldots \cdot {{|R_{p-1}|}\choose{m}}}$ equivalence classes. It thus induces a coloring of $\chi' : R_p^m \to \cC'_p$ with one color for each equivalence class.
	
	By virtue of Lemma \ref{lemma:RamseyBasic:One}, we can construct a subset $Q_p \subseteq R_p$ with $n$ elements such that all ascending $m$-tuples $\bar{r} \in Q_p^m$ are colored identically by $\chi'$.
	
	Let the coloring $\chi''$ be defined by $\chi''(\bar{r}_1, \ldots, \bar{r}_{p-1}) := \chi(\bar{r}_1, \ldots, \bar{r}_{p-1}, \bar{s})$ for some fixed ascending $m$-tuple $\bar{s} \in Q_{p}^m$. 
	By induction, we find subsets $Q_1 \subseteq R_1, \ldots, Q_{p-1} \subseteq R_{p-1}$, each containing $n$ elements, such that for all ascending $m$-tuples $\bar{r}_1 \in R_1^m, \ldots, \bar{r}_{p-1} \in R_{p-1}^m$ the colors $\chi''(\bar{r}_1, \ldots, \bar{r}_{p-1})$ are identical.
	
	But then the definition of $\chi''$ and $\chi'$ entail that the sets $Q_1, \ldots, Q_p$ satisfy the requirements posed by the lemma.
\end{proof}

Recall that we write $[k]$ to address the set $\{1, \ldots, k\}$ for any positive integer $k > 0$.

\begin{lemma}\label{lemma:RamseyPermute:One}
	Let $n, m, p > 0$ be positive integers, let $\kappa \geq 0$ be a nonnegative integer and let $\chi: \Real^m \to \cC$ be an arbitrary coloring.
	Let $R_1, \ldots, R_p$ be \emph{sufficiently large} but finite subsets of $\Real$.
	Let $q_1, \ldots, q_\kappa$ be fixed reals.
	Let $\varrho : [m] \to [p+\kappa]\times[m]$ be some mapping such that $\varrho(i) = \<k,\ell\>$ with $k > p$ implies $\ell = 1$.
	
	There exist subsets $Q_1 \subseteq R_1, \ldots, Q_p \subseteq R_p$, each of cardinality $n$, such that for all ascending tuples 
		\begin{align*}
			\bar{r}_1 = \<r_{\<1,1\>}, \ldots, r_{\<1,m\>}\> &\in Q_1^{m}\\
				\vdots\quad&~\\
			\bar{r}_p = \<r_{\<p,1\>}, \ldots, r_{\<p,m\>}\> &\in Q_p^{m}
		\end{align*}
	and the reals $r_{\<p+1,1\>} := q_1, \dots, r_{\<p+\kappa,1\>} := q_\kappa$	
	the colors $\chi(\bar{r}_{\varrho(1)}, \ldots, \bar{r}_{\varrho(m)})$ are the same.	
\end{lemma}
\begin{proof}
	We again assume $n \geq m$.
	We define a new coloring $\chi' : \Real^{m p} \to \cC$ by
		\[\chi'(r_{\<1,1\>}, \ldots, r_{\<1,m\>}, \ldots, r_{\<p,1\>}, \ldots, r_{\<p,m\>}) := \chi(r_{\varrho(1)}, \ldots, r_{\varrho(m)})\]
	for every $m p$-tuple $\<\bar{r}_1, \ldots, \bar{r}_p\> \in R_1^m \times \ldots \times R_p^m$ with ascending $\bar{r}_1, \ldots, \bar{r}_p$. 
	By Lemma~\ref{lemma:RamseyBasic:Two}, there exist subsets $Q_1 \subseteq R_1, \ldots, Q_p \subseteq R_p$, each with $n$ elements, such that for all ascending tuples $\bar{r}_1 \in Q_1^m, \ldots, \bar{r}_p \in Q_p^m$	 the colors $\chi'(\bar{r}_1, \ldots, \bar{r}_p)$ are the same.
	By definition of $\chi'$, the sets $Q_1, \ldots, Q_p$ satisfy the requirements of the lemma.
\end{proof}

Now we have the right tools at hand to prove Lemma~\ref{lemma:RamseyPermute:Two}

\begin{lemma*}
	Let $n, m, p > 0$ be positive integers, let $K \geq 0$ be a nonnegative integer and let $\chi: \Real^m \to \cC$ be an arbitrary coloring.
	Let $R_1, \ldots, R_p$ be \emph{sufficiently large} but finite subsets of $\Real$.
	Let $q_1, \ldots, q_K$ be fixed reals.
	Let $\varrho_1, \ldots, \varrho_L$ be some enumeration of all mappings $\varrho_j : [m] \to [p+K]\times[m]$ for which $\varrho_{j}(i) = \<k,\ell\>$ with $k > p$ entails $\ell = 1$.
	There exist subsets $Q_1 \subseteq R_1, \ldots, Q_p \subseteq R_p$, each of cardinality $n$, such that for all ascending tuples 
		$\bar{r}_1, \bar{r}'_1 \in Q_1^m, \ldots, \bar{r}_p, \bar{r}'_p \in Q_p^m$
	and the reals $r_{\<p+1,1\>} := q_1, \dots, r_{\<p+K,1\>} := q_K$	
	and every index $j$, $1\leq j\leq L$, we have \\
		\centerline{$\chi\bigl(r_{\varrho_j(1)}, \ldots, r_{\varrho_j(m)}\bigr) = \chi\bigl(r'_{\varrho_j(1)}, \ldots, r'_{\varrho_j(m)}\bigr)$.}
\end{lemma*}
\begin{proof}
	We again assume $n \geq m$.
	We construct sequences of subsets $S_{\ell,0} \supseteq \ldots \supseteq S_{\ell,L}$ for every $\ell$, $1\leq \ell \leq p$, such that
	\begin{itemize}
		\item $S_{\ell,0} = R_\ell$, and
		\item $S_{\ell, j+1} \subseteq S_{\ell, j}$ is a subset of \emph{sufficient cardinality} that is constructed by application of Lemma \ref{lemma:RamseyPermute:One} for $\varrho := \varrho_{j+1}$, i.e.\  for all ascending tuples 
		\begin{align*}
			\<s_{\<1,1\>}, \ldots, s_{\<1,m\>}\> &\in S_{1, j+1}^m\\
				\vdots\quad&~\\
			\<s_{\<p,1\>}, \ldots, s_{\<p,m\>}\> &\in S_{p, j+1}^m
		\end{align*}	
	 the colors $\chi(\bar{s}_{\varrho_{j+1}(1)}, \ldots, \bar{s}_{\varrho_{j+1}(m)})$ are the same.
	\end{itemize}
	Then the sets $S_{1, L}, \ldots, S_{p, L}$ are the sought $Q_1, \ldots, Q_p$.
\end{proof}

\subsection{Details Concerning Section~\ref{section:DecidabilityBsrSla}}

\subsubsection*{Proof of Lemma~\ref{lemma:ExistenceOfUniformModeslForBsrSla}}

\begin{lemma*}
	Let $\lambda$ be the maximal number of distinct base-sort variables in any single clause in $N$ but at least $m$, i.e.\ $\lambda := \max\bigl( \{m\} \cup \bigl\{ |\vars(C) \cap V_\cR| \bigm| C\in N \bigr\} \bigr)$. 
	Let $\cA$ be a  model of $N$.
	Let $J_0, \ldots, J_\kappa$ be an enumeration of all open intervals in $\JA$ so that $J_0 < \ldots < J_\kappa$. 
	Suppose we are given a collection of finite sets $Q_0, \ldots, Q_\kappa$ possessing the following properties, 
		\begin{enumerate}[label=(\roman{*}), ref=(\roman{*})]
			\item\label{enum:UniformModelConstruction:One} $Q_i \subseteq J_i$ and $|Q_i| = \lambda$ for every $i$, $0 \leq i \leq \kappa$.
			\item\label{enum:UniformModelConstruction:Two} Let $Q := \bigcup_i Q_i \cup \{c^\cA \mid c \in \bconsts(N) \}$. 
				For all $\JA$-equivalent $m$-tuples $\bq, \bq' \in Q^m$ we have $\chi_\cA(\bq) = \chi_\cA(\bq')$.
		\end{enumerate}
	Then we can construct a  model $\cB$ of $N$ that is $\JB$-uniform and that interprets the free sort $\cS$ as a finite set.
	Moreover, $\cB$ interprets all constant symbols in $N$ exactly as $\cA$ does.	
\end{lemma*}
\begin{proof}~
	\begin{description}
		\item \underline{Claim I:} Let $\mu$, $1\leq \mu \leq, \lambda$ be a positive integer.  
			For each of the finitely many equivalence classes in $\Real^\mu/_{\sim_\JA}$, we find a representative lying in $Q^\mu$.
			
		\item \underline{Proof:} Given an equivalence class $[\br]_{\sim_\JA} \in \Real^\mu/_{\sim_\JA}$, we define the following ascending sequences for every $i$, $0\leq i\leq \kappa$,
			\begin{itemize}
				\item $s_{i,1} < \ldots < s_{i,k_i}$, where the values $s_{i,j}$ are the reals in $\br$ that stem from $J_i$, enumerated in ascending order, and
				\item $q_{i,1} < \ldots < q_{i,\lambda}$, which comprises all reals in $Q_i$ in ascending order.
			\end{itemize}
			In every $Q_i \subseteq J_i$ we find $\lambda \geq \mu \geq k_i$ distinct reals.
			
			We can now construct a tuple $\bq' \in [\br]_{\sim_\JA} \cap Q^\mu$ by setting
				\[q'_\ell := \begin{cases}
						c^\cA		&\text{if $r_\ell = c^\cA$ for some $c \in \bconsts(N)$},\\
						q_{i,j}		&\text{if $r_\ell = s_{i,j}$ for some $i$, $0 \leq i \leq \kappa$, and some $j$, $1 \leq j \leq k_i$,}
					\end{cases}
				\]
			for every $\ell$, $1 \leq \ell \leq \mu$.	
			Clearly, $\br$ and $\bq'$ are $\JA$-equivalent.	\hfill $\Diamond$			
	\end{description}

	We construct the  interpretation $\cB$ as follows, where $\hcS$ denotes the set $\{ \fa \in \cS^\cA \mid \text{$\fa = c^\cA$ for some}$ $c \in \fconsts(N) \}$:
	\begin{itemize}
		\item $\cS^\cB := \hcS$,
		\item for every constant symbol $c$ occurring in $N$ we set $c^\cB := c^\cA$,
		\item for every uninterpreted predicate symbol $P$ occurring in $N$ and for all tuples $\ba \in \hcS^{m'}$ and $\br \in \Real^m$ we pick some tuple $\bq \in Q^m$ which is $\JA$-equivalent to $\br$, and we define $P^\cB$ so that
			\[ \<\ba, \br\> \in P^\cB \quad\text{if and only if}\quad \<\ba,\bq\> \in P^\cA ~.\]
	\end{itemize}
	
	\begin{description}
		\item \underline{Claim II:} The  interpretation $\cB$ is $\JB$-uniform.
		\item \underline{Proof:} By construction of $\cB$ and by requirement \ref{enum:UniformModelConstruction:Two}.	\hfill $\Diamond$
	\end{description}
	
	We next show $\cB\models N$. 
	Consider any clause $C = \Lambda \;\|\; \Gamma \to \Delta$ in $N$ and let $\beta$ be any variable assignment ranging over $\cS^\cB \cup \Real$. 
	Starting from $\beta$, we derive a special variable assignment $\hbeta_C$ as follows. 
	Let $x_1, \ldots, x_{\lambda_C}$ be an enumeration of all base-sort variables in $C$.
	By Claim~I, there is some tuple $\<q_1, \ldots, q_{\lambda_C}\> \in Q^{\lambda_C}$ such that $\<q_1, \ldots, q_{\lambda_C}\> \sim_\JA \bigl\< \beta(x_1), \ldots, \beta(x_{\lambda_C}) \bigr\>$.
	We define $\hbeta_C(x_i) := q_i$ for every $i$, $1 \leq i \leq \lambda_C$.
	For all other base-sort variables, $\hbeta_C$ can be defined arbitrarily.
	For every free-sort variable $u$ we set $\hbeta_C(u) := \beta(u)$.
	We observe 
	\centerline{$(*)\qquad
		\bigl\< \beta(x_1), \ldots, \beta(x_{\lambda_C}) \bigr\> \sim_\JB \bigl\< \hbeta_C(x_1), \ldots, \hbeta_C(x_{\lambda_C}) \bigr\>$.}

	As $\cA$ is a  model of $N$, we get $\cA,\hbeta_C \models C$. By case distinction on why $\cA,\hbeta_C \models C$ holds, we can infer $\cB,\beta\models C$.
		\begin{description}
			\item Case $\cA, \hbeta_C \not\models t\rel t'$ for some atomic LA constraint $t\rel t'$ in $\Lambda$, where $t, t'$ are constant symbols or base-sort variables. 
				Since $\cB$ and $\cA$ interpret constant symbols in the same way and due to $(*)$, we conclude $\cB, \beta \not\models t \rel t'$.

			\item Case $\cA, \hbeta_C \not\models t\approx t'$ for some free-sort equation $t\approx t' \in \Gamma$. 
				In this case, $t$ and $t'$ are either variables or constant symbols of the free sort, which means they do not contain subterms of the base sort. 
				Since $\cB$ and $\cA$ behave identical on free-sort constant symbols and $\beta(u) = \hbeta_C(u)$ for every variable $u\in V_\cS$, we have $\cB, \beta \not\models t\approx t'$.

			\item Case $\cA, \hbeta_C \models t\approx t'$ for some $t\approx t' \in \Delta$. Analogous to the above case, we get $\cB, \beta \models t\approx t'$.
			
			\item Case $\cA, \hbeta_C \not\models P(t'_1, \ldots, t'_{m'}, t_1, \ldots, t_m)$ for some non-equational atom\\ $P(t'_1, \ldots, t'_{m'}, t_1, \ldots, t_m) \in \Gamma$. 
				This translates to \\
					\centerline{$\bigl\< \cA(\hbeta_C)(t'_1), \ldots, \cA(\hbeta_C)(t'_{m'}), \cA(\hbeta_C)(t_1), \ldots, \cA(\hbeta_C)(t_m)\bigr\> \not\in P^\cA$.}
				By definition of $\hbeta_C$, we have $\cA(\hbeta_C)(t_j) \in Q$ for every $j$, $1 \leq j \leq m$.
				Therefore, and by construction of $\cB$, \\
					\centerline{$\bigl\< \cA(\hbeta_C)(t'_1), \ldots, \cA(\hbeta_C)(t'_{m'}), \cA(\hbeta_C)(t_1), \ldots, \cA(\hbeta_C)(t_m) \bigr\> \not\in P^\cB$.}
				We observe the following properties:
				\begin{itemize}
					\item We have $\cA(\hbeta_C)(t'_j) = \cB(\beta)(t'_j)$ for every $j$, $1 \leq j \leq m'$, due to the definition of $\cB$ and $\hbeta_C$.
					\item Since $\cA$ and $\cB$ interpret constant symbols in the same way, we get $\cA(\hbeta_C)(t_j) = \cB(\hbeta_C)(t_j)$ for every $j$, $1 \leq j \leq m$.
					\item The definition of $\hbeta_C$ entails that $\bigl\<\cB(\hbeta_C)(t_1), \ldots, \cB(\hbeta_C)(t_m)\bigr\>$ and\\
						 $\bigl\<\cB(\beta)(t_1), \ldots, \cB(\beta)(t_m)\bigr\>$ are $\JB$-equivalent. 
				\end{itemize}				
				The first two observations imply\\
					\centerline{$\bigl\< \cB(\beta)(t'_1), \ldots, \cB(\beta)(t'_{m'}), \cB(\hbeta_C)(t_1), \ldots, \cB(\hbeta_C)(t_m) \bigr\> \not\in P^\cB$.}
				Due to this result and the fact that $\cB$ is $\JB$-uniform (Claim II), the third observation leads to
					$\bigl\< \cB(\beta)(t'_1), \ldots, \cB(\beta)(t'_{m'}), \cB(\beta)(t_1), \ldots, \cB(\beta)(t_m) \bigr\> \not\in P^\cB$.
					
				Put differently, we have $\cB, \beta \not\models P(t'_1, \ldots, t'_{m'}, t_1, \ldots, t_m)$.
													
			\item Case $\cA, \hbeta_C \models P(t'_1, \ldots, t'_{m'}, t_1, \ldots, t_m)$ for some non-equational atom
				$P(t'_1, \ldots, t'_{m'}, t_1, \ldots, t_m)\in\Delta$. 
				Analogous to the previous case we can infer $\cB, \beta \models P(t'_1, \ldots,$ $t'_{m'}, t_1, \ldots, t_m)$.
		\end{description}		
	Altogether, we have shown $\cB\models N$.	
\end{proof}

\subsubsection*{Proof of Lemma~\ref{lemma:IndistuingishableReals}}

As an auxiliary result, we first show a correspondence between the equivalence classes with respect to\ $\sim_\JA$ and mappings \mbox{$\varrho : [m] \to [|\JA|] \times [m]$}.
\begin{lemma}\label{lemma:Correspondence}
	Let $\cA$ be an  interpretation.
	Let $\{q_1\}, \ldots, \{q_\kappa\}$ be an enumeration of all point intervals in $\JA$ such that $q_1 < \ldots < q_\kappa$
	and let $J_1, \ldots, J_{\kappa+1}$ be an enumeration of all open intervals in $\JA$ such that $J_1 < \ldots < J_{\kappa+1}$. 
	Let $S \in \Real^m/_{\sim_\JA}$ be any equivalence class with respect to $\sim_\JA$. There is a mapping $\varrho : [m] \to [|\JA|] \times [m]$ such that
	\begin{enumerate}[label=(\roman{*}), ref=(\roman{*})]
		\item\label{enum:lemmaCorrespondence:One} whenever $\varrho(i) = \<k,\ell\>$ with $k > \kappa + 1$ then $\ell = 1$, and
		\item\label{enum:lemmaCorrespondence:Two} for all ascending tuples 
			\begin{align*}
				\br_1 &= \<r_{\<1,1\>}, \ldots, r_{\<1,m\>}\> \in J_1^m, \\
				&\;\;\vdots \\
				\br_{\kappa+1} &= \<r_{\<\kappa+1,1\>}, \ldots, r_{\<\kappa+1,m\>}\> \in J_{\kappa+1}^m, \\
				\br_{\kappa+2} &= \<r_{\<\kappa+2,1\>}\> = \<q_1\> \\
				&\;\;\vdots \\
				\br_{2K+1} &= \<r_{\<2K+1,1\>}\> = \<q_\kappa\>
			\end{align*}
			we have $\<r_{\varrho(1)}, \ldots, r_{\varrho(m)}\> \in S$, and
			
		\item\label{enum:lemmaCorrespondence:Three} for every tuple $\<s_1, \ldots, s_m\> \in S$ there exist ascending tuples $\br_1, \ldots, \br_{2K+1}$ defined as in \ref{enum:lemmaCorrespondence:Two} such that $\<s_1, \ldots, s_m\> = \<r_{\varrho(1)}, \ldots, r_{\varrho(m)}\>$.
	\end{enumerate}
\end{lemma}
\begin{proof}
	Let $\bs'$ be some representative taken from $S$, i.e.\ $S = [\bs']_{\sim_\JA}$. 
	Given $\bs'$, we construct $2K+1$ possibly empty sequences $\bs''_k := \<s''_{k,1}, s''_{k,2}, \ldots\>$, 
	such that every $\bs''_k$ with $k \leq \kappa+1$ lists all elements of $\bs'$ in ascending order that lie in $J_k$, 
	and every $\bs''_k$ with $k > \kappa+1$ contains exactly the value $q_{k-\kappa-1}$.
	We construct the mapping $\varrho$ in such a way that $\varrho(i) = \<k, \ell\>$ holds if and only if $s'_i = s''_{k,\ell}$.
	
	Let $\br_1, \ldots, \br_{2K+1}$ be tuples of reals chosen in accordance with requirement \ref{enum:lemmaCorrespondence:Two}. It is easy to verify that $\br_\varrho := \<r_{\varrho(1)}, \ldots, r_{\varrho(m)}\>$ is $\JA$-equivalent to $\bs'$, i.e.\ $\br_\varrho$ belongs to $S$.
	
	In order to show \ref{enum:lemmaCorrespondence:Three}, we construct the tuples $\br_1, \ldots, \br_{2K+1}$ from $\<s_1, \ldots, s_m\>$ in the same way we have constructed the $\bs''_k$ from $\bs'$ when constructing $\varrho$ in the beginning of this proof.
	In addition, we pad them with suitable values from the respective intervals $J_k$ to reach the length $m$ for every tuple.
\end{proof}

We can now prove Lemma~\ref{lemma:IndistuingishableReals}.

\begin{lemma*}
	Let $\cA$ be an  interpretation.
	Let $\{q_1\}, \ldots, \{q_\kappa\}$ be an enumeration of all point intervals in $\JA$ such that $q_1 < \ldots < q_\kappa$
	and let $J_1, \ldots, J_{\kappa+1}$ be an enumeration of all open intervals in $\JA$ such that $J_1 < \ldots < J_{\kappa+1}$. 
	Let $\lambda$ be a positive integer.
	There is a collection of finite sets $Q_1, \ldots, Q_{\kappa+1}$ such that the following requirements are met.
	\begin{enumerate}[label=(\roman{*}), ref=(\roman{*})]
		\item\label{enum:lemmaIndistuingishableReals:One} For every $i$, $1 \leq i \leq \kappa+1$, it holds $Q_i \subseteq J_i$ and $|Q_i| = \lambda$.
		\item\label{enum:lemmaIndistuingishableReals:Two} Let $Q := \bigcup_i Q_i \cup \{q_1, \ldots, q_\kappa\}$. For all $\JA$-equivalent $m$-tuples $\bs, \bs' \in Q^m$ we have $\chi_\cA(\bs) = \chi_\cA(\bs')$.
	\end{enumerate}
\end{lemma*}
\begin{proof}
	Let the sets $Q_1, \ldots, Q_{\kappa+1}$ be the $Q_1, \ldots, Q_p$ that we obtain by virtue of Lemma \ref{lemma:RamseyPermute:Two} when we set $n := \lambda$, $p := \kappa+1$, $\chi := \chi_\cA$, $R_1 := J_1, \ldots, R_{\kappa+1} := J_{\kappa+1}$.

	Requirement \ref{enum:lemmaIndistuingishableReals:One} is obviously satisfied for $Q_1, \ldots, Q_{\kappa+1}$.
	
	By Lemma \ref{lemma:Correspondence}, the equivalence class to which any two given $\JA$-equivalent tuples $\bs, \bs'$ belong corresponds to some mapping $\varrho : [m] \to [2K+1]\times[m]$. Part \ref{enum:lemmaCorrespondence:Two} of Lemma \ref{lemma:Correspondence} states that $\bs$ can be written in the form $\<r_{\varrho(1)}, \ldots, r_{\varrho(m)}\>$ for appropriate values $r_{\<k,\ell\>}$ and $\bs'$ can be represented in the form $\<r'_{\varrho(1)}, \ldots, r'_{\varrho(m)}\>$ for appropriate $r'_{\<k,\ell\>}$. We then know by Lemma  \ref{lemma:RamseyPermute:Two} that $\chi_\cA(\bs) = \chi_\cA(\<r_{\varrho(1)}, \ldots, r_{\varrho(m)}\>) = \chi_\cA(\<r'_{\varrho(1)}, \ldots, r'_{\varrho(m)}\>) = \chi_\cA(\bs')$.
\end{proof}

\subsection{Details Concerning Section~\ref{section:DecidabilityBsrBd}}

\subsubsection*{Proof of Lemma~\ref{lemma:ExistenceOfUniformModeslForBsrBd}}

\begin{lemma*}
	Let $\lambda := \max\bigl( \{m\} \cup \bigl\{ |\vars(C) \cap V_\cR| \bigm| C \in N \bigr\} \bigr)$. 
	Let $\cA$ be a  model of $N$ and let $\kappa := \max\bigl( \{1\} \cup \{ |c| \bigm| c \in \bconsts(N) \} \bigr)$.
	Suppose we are given a finite set $Q \subset [0,1)$ of cardinality $\lambda+1$ such that $0 \in Q$ and for all tuples $\br, \bs \in \hQ^m$,
		$\br \simeq_\kappa \bs$ entails $\chi_\cA(\br) = \chi_\cA(\bs)$,
	where \\
		\centerline{$\hQ := \bigl\{ q + k \bigm| q \in Q \text{ and } k \in \{-\kappa-1, \ldots, 0, \ldots, \kappa \} \bigr\}$.}
	Then we can construct a model $\cB$ of $N$ that is $\hsimeq_\kappa$-uniform and that interprets the free sort $\cS$ as a finite set.
\end{lemma*}
\begin{proof}
	The construction of $\cB$ from $\cA$ is similar to the construction of uniform models outlined in the proof of Lemma~\ref{lemma:ExistenceOfUniformModeslForBsrSla}.
	
	\begin{description}
		\item \underline{Claim I:} Let $\mu$ be a positive integer with $1 \leq \mu \leq \lambda$.  
			For each of the finitely many equivalence classes $S \in \Real^\mu/_{\hsimeq_\kappa}$ and every $\br \in S$, there is some $\bq \in S \cap \hQ^\mu$ such that $\br \hsimeq_\kappa \bq$ and for all $i_1, i_2, i_3$ with $r_{i_1} < -\kappa$ and $r_{i_2} > \kappa$ and $-\kappa \leq r_{i_3} \leq \kappa$ we have $\Fr(q_{i_1}) < \Fr(q_{i_2}) < \Fr(q_{i_3})$.
		\item \underline{Proof:} 
			Let $i_1, i_2, \ldots$ be all the indices from $\{1, \ldots, \mu\}$ for which we have $r_{i_j} > \kappa$ for every $j$.
			Analogously, let $\ell_1, \ell_2, \ldots$ be all the indices from $\{1, \ldots, \mu\}$ such that $r_{\ell_j} < -\kappa$ holds for every $j$.
			We define the real \\
				\centerline{$\delta := \min \bigl\{ \Fr(r_i) \bigm| \text{$-\kappa \leq r_i \leq \kappa$ and $\Fr(r_i) > 0$ and $1 \leq i \leq m$} \bigr\} \cup \bigl\{ \frac{1}{2} \bigr\}$.}
			There must be some integer $t$ for which we get $0 < \frac{1}{t}r_{i_j} < \frac{1}{2}\delta$ and $-\frac{1}{2}\delta < \frac{1}{t}r_{\ell_j} < 0$ for every $j$.
			Let $\br'$ be the tuple that we obtain from $\br$ by replacing every $r_{i_j}$ with $\frac{1}{t}r_{i_j} + \frac{1}{2}\delta + \kappa$ and every $r_{\ell_j}$ with $\frac{1}{t}r_{\ell_j} + \frac{1}{2}\delta - \kappa$.
			By construction, we observe $\br' \in (-\kappa-1, \kappa+1)^\mu$ and $\br \hsimeq_\kappa \br'$.
			Moreover, we have $\frac{1}{2}\delta < \Fr(\br'_{i_j}) < \delta$ and $0 < \Fr(\br'_{\ell_j}) < \frac{1}{2}\delta$ for every $j$.
		
			Next, we define the following ascending sequences
			\begin{itemize}
				\item $s'_0 < s'_1 < \ldots < s'_k$, where $s'_0 = 0$ and the values $s'_j$, $j \geq 1$, are the strictly positive fractional parts in ascending order that occur in $\Fr(\br')$, and
				\item $q'_0 < q'_1 < \ldots < q'_\lambda$, which comprises all reals in $Q$ in ascending order, including $q'_0 = 0$.
			\end{itemize}
		
			We can now construct a tuple $\bq \in S \cap \hQ^\mu$ by setting $q_\ell := \floor{r'_\ell} + q'_j$ for $j$ such that $\Fr(r'_\ell) = s'_j$.
			
			Clearly, $\br'$ and $\bq$ are $\simeq_\kappa$-equivalent.
			Since $\simeq_\kappa$ is a refinement of $\hsimeq_\kappa$ on the subspace $(-\kappa - 1, \kappa + 1)^\mu$, this entails $\br \hsimeq_\kappa \bq$. 
			\hfill $\Diamond$			
	\end{description}
	
	Let $\hcS$ denote the set $\{ \fa \in \cS^\cA \mid \text{$\fa = c^\cA$ for some $c \in \fconsts(N)$} \}$.
	The  interpretation $\cB$ can be constructed as follows:
	\begin{itemize}
		\item $\cS^{\cB} := \hcS$,
		\item for every constant symbol $c$ occurring in $N$ we set $c^{\cB} := c^\cA$,
		\item for every uninterpreted predicate symbol $P$ occurring in $N$ and for all tuples $\ba \in \hcS^{m'}$ and $\br \in \Real^m$ we pick some tuple $\bq \in \hQ^m$ in accordance with Claim~I---i.e.\ $\bq$ satisfies $\br \hsimeq_\kappa \bq$---and define $P^{\cB}$ in such a way that
			\[ \<\ba, \br\> \in P^{\cB} \quad\text{if and only if}\quad \<\ba,\bq\> \in P^\cA ~.\]
	\end{itemize}
	
	\begin{description}
		\item \underline{Claim II:} The  interpretation $\cB$ is $\hsimeq_\kappa$-uniform.
		\item \underline{Proof:}
			Let $\br^1, \br^2 \in \Real^m$ be two $\hsimeq_\kappa$-equivalent tuples.
			By Claim~I, there exist two tuples $\bq^1, \bq^2$ such that $\bq^1 \hsimeq_\kappa \br^1$ and $\bq^2 \hsimeq_\kappa \br^2$.
			Clearly, by transitivity and symmetry of $\hsimeq_\kappa$, we have $\bq^1 \hsimeq_\kappa \bq^2$.
			Even stronger, we can show $\bq^1 \simeq_\kappa \bq^2$.
			Suppose, $\bq^1 \not\simeq_\kappa \bq^2$.
			We observe the following properties, which follow from $\bq^1 \hsimeq_\kappa \bq^2$:
			\begin{itemize}
				\item  $\floor{\bq^1} = \floor{\bq^2}$ and $\ceil{\bq^1} = \ceil{\bq^2}$.
				\item For all $i, j$, $1 \leq i, j \leq m$, for which $-\kappa \leq q^1_i, q^1_j \leq \kappa$, 
					we have $\Fr(q^1_i) \leq \Fr(q^1_j)$ if and only if $\Fr(q^2_i) \leq \Fr(q^2_j)$.
				\item For all $i, j$, $1 \leq i, j \leq m$, for which $\kappa < q^1_i, q^1_j$ or $q^1_i, q^1_j < -\kappa$, 
					we have $q^1_i \leq q^1_j$ if and only if $q^2_i \leq q^2_j$.
					Because of $\bq^1, \bq^2 \in (-\kappa-1, \kappa+1)^m$, 
					we even obtain $\Fr(q^1_i) \leq \Fr(q^1_j)$ if and only if $\Fr(q^2_i) \leq \Fr(q^2_j)$.					
			\end{itemize}
			Hence, our assumption $\bq^1 \not\simeq_\kappa \bq^2$ entails that there are two indices $i,j$ such that
				$\Fr(q^1_i) \leq \Fr(q^1_j)$ and $\Fr(q^2_i) > \Fr(q^2_j)$,
			and one of the following cases applies:
				\begin{enumerate}[label=(\arabic{*}), ref=(\arabic{*})]
					\item\label{enum:proofExistenceOfUniformModeslForBsrBd:I} $q^1_i, q^2_i > \kappa$ and $-\kappa \leq q^1_j, q^2_j \leq \kappa$, or
					\item\label{enum:proofExistenceOfUniformModeslForBsrBd:II} $q^1_i, q^2_i > \kappa$ and $q^1_j, q^2_j < -\kappa$, or
					\item\label{enum:proofExistenceOfUniformModeslForBsrBd:III} $-\kappa \leq q^1_i, q^2_i \leq \kappa$ and $q^1_j, q^2_j >\kappa$, or
					\item\label{enum:proofExistenceOfUniformModeslForBsrBd:IV} $-\kappa \leq q^1_i, q^2_i \leq \kappa$ and $-\kappa < q^1_j, q^2_j$, or
					\item\label{enum:proofExistenceOfUniformModeslForBsrBd:V} $q^1_i, q^2_i < -\kappa$ and $-\kappa \leq q^1_j, q^2_j$.
				\end{enumerate}
			\begin{description}
				\item Ad~\ref{enum:proofExistenceOfUniformModeslForBsrBd:I}.
					By Claim~I, we have $\Fr(q^1_i) < \Fr(q^1_j)$ and $\Fr(q^2_i) < \Fr(q^2_j)$.
				\item Ad~\ref{enum:proofExistenceOfUniformModeslForBsrBd:II}.
					By Claim~I, we have $\Fr(q^1_j) < \Fr(q^1_i)$ and $\Fr(q^2_j) < \Fr(q^2_i)$.
				\item Ad~\ref{enum:proofExistenceOfUniformModeslForBsrBd:III}.
					By Claim~I, we have $\Fr(q^1_i) < \Fr(q^1_j)$ and $\Fr(q^2_i) < \Fr(q^2_j)$.
				\item Ad~\ref{enum:proofExistenceOfUniformModeslForBsrBd:IV}.
					By Claim~I, we have $\Fr(q^1_j) < \Fr(q^1_i)$ and $\Fr(q^2_j) < \Fr(q^2_i)$.
				\item Ad~\ref{enum:proofExistenceOfUniformModeslForBsrBd:V}.
					By Claim~I, we have $\Fr(q^1_i) < \Fr(q^1_j)$ and $\Fr(q^2_i) < \Fr(q^2_j)$.
			\end{description}
			Since all cases lead to a contradiction, we must have $\bq^1 \simeq_\kappa \bq^2$.
		
			Because of $\bq^1, \bq^2 \in \hQ^m$ and due to our assumptions regarding $Q$ and $\hQ^m$, we have $\chi_\cA(\bq^1) = \chi_\cA(\bq^2)$.
			Moreover, by construction of $\cB$, we have $\chi_\cB(\br^1) = \chi_\cA(\bq^1)$ and $\chi_\cB(\br^2) = \chi_\cA(\bq^2)$.
			Consequently, $\chi_\cB(\br^1) = \chi_\cB(\br^2)$.
			\strut\hfill $\Diamond$
	\end{description}
	
	We next show $\cB\models N$. 
	Consider any clause $C = \Lambda \;\|\; \Gamma \to \Delta$ in $N$ and let $\beta$ be any variable assignment ranging over $\cS^{\cB} \cup \Real$. 
	Starting from $\beta$, we derive a special variable assignment $\hbeta_C$ as follows. 
	Let $x_1, \ldots, x_\ell$ be an enumeration of all base-sort variables in $C$.
	By Claim~I, there exists some tuple $\bq := \<q_1, \ldots, q_\ell\>$ such that $\<q_1, \ldots, q_\ell\> \hsimeq_\kappa \bigl\< \beta(x_1), \ldots, \beta(x_\ell) \bigr\>$ and $\bq \in \hQ^\ell$.	
	We define $\hbeta_C(x_i) := q_i$ for every $i$, $1 \leq i \leq \ell$.
	Hence, we have
	\begin{equation*}(*)\qquad
		\bigl\<\hbeta_C(x_1), \ldots, \hbeta_C(x_\ell)\bigr\> \hsimeq_\kappa \bigl\<\beta(x_1), \ldots, \beta(x_\ell)\bigr\> ~.
	\end{equation*}	
	For all other base-sort variables $y \not\in \{x_1, \ldots, x_\ell\}$, $\hbeta_C(y)$ can be defined arbitrarily.
	For every free-sort variable $u$ we set $\hbeta_C(u) := \beta(u)$.

	As $\cA$ is a  model of $N$, we know $\cA,\hbeta_C \models C$. By case distinction on why $\cA,\hbeta_C \models C$ holds, we may use this result to obtain $\cB,\beta\models C$.
		\begin{description}
			\item Case $\cA, \hbeta_C \not\models x\rel c$ for some constraint $x\rel c$ in $\Lambda$. 
				Hence, $\beta_C(x) \not\rel c$.
				Due to $(*)$, the assumption $|c| \leq \kappa$, and the definition of $\hsimeq_\kappa$, we know that $\hbeta_C(x) \rel c$ holds if and only if $\beta(x) \rel c$ holds. 
				Consequently, we get $\beta(x) \not\rel c$ and thus $\cB, \beta \not\models x \rel c$.
									
			\item Case $\cA, \hbeta_C \not\models x \rel y$ for some $x \rel y$ in $\Lambda$. 
				By $(*)$ and the definition of $\hsimeq_\kappa$, we know that $\hbeta_C(x) \rel \hbeta_C(y)$ if and only if $\beta(x) \rel \beta(y)$. 
				Consequently, we get $\cB, \beta \not\models x \rel y$.

			\item Case $\cA, \hbeta_C \not\models x - y \rel c$ for some constraint $x - y \rel c$ in $\Lambda$. 
				By definition of \BsrBd\ clause sets, $\Lambda$ must also contain constraints $c_x \leq x$, $x \leq d_x$, $c_y \leq y$, and $y \leq d_y$ for certain constants $c_x, d_x,  c_y, d_y$ whose absolute value is at most $\kappa$.
				If one of these constraints is violated by $\hbeta_C$, then the first case applies.
				
				If all of these constraints are satisfied by $\hbeta_C$, then, by $(*)$, they are also satisfied by $\beta$.
				Moreover, $(*)$ and the definition of $\hsimeq_\kappa$, entail $\floor{\hbeta_C(x)} = \floor{\beta(x)}$, $\floor{\hbeta_C(y)} = \floor{\beta(y)}$, $\ceil{\hbeta_C(x)} = \ceil{\beta(x)}$, $\ceil{\hbeta_C(y)} = \ceil{\beta(y)}$, $\Fr(\hbeta_C(x)) \leq \Fr(\hbeta_C(y))$ if and only if $\Fr(\beta(x)) \leq \Fr(\beta(y))$, and $\Fr(\hbeta_C(x)) \geq \Fr(\hbeta_C(y))$ if and only if $\Fr(\beta(x)) \geq \Fr(\beta(y))$.
				Hence, the following two observations hold:
					\begin{align*}
						\floor{\hbeta_C(x) - \hbeta_C(y)}
							&= \floor{\hbeta_C(x)} - \floor{\hbeta_C(y)} + \bigl\lfloor \Fr(\hbeta_C(x)) - \Fr(\hbeta_C(y)) \bigr\rfloor \\
							&= \floor{\beta(x)} - \floor{\beta(y)} + \bigl\lfloor \Fr(\beta(x)) - \Fr(\beta(y)) \bigr\rfloor \\
							&= \floor{\beta(x) - \beta(y)}
					\end{align*}		
				and	
					\begin{align*}
						\ceil{\hbeta_C(x) - \hbeta_C(y)}
							&= \ceil{\hbeta_C(x)} - \ceil{\hbeta_C(y)} + \bigl\lceil \Fr(\hbeta_C(x)) - \Fr(\hbeta_C(y)) \bigr\rceil \\
							&= \ceil{\beta(x)} - \ceil{\beta(y)} + \bigl\lceil \Fr(\beta(x)) - \Fr(\beta(y)) \bigr\rceil \\
							&= \ceil{\beta(x) - \beta(y)} ~.
					\end{align*}		
				Consequently, we have $\hbeta_C(x) - \hbeta_C(y) \rel c$ if and only if $\beta(x) - \beta(y) \rel c$.
				In other words, $\cA, \beta \not\models x - y \rel c$.

			\item Case $\cA, \hbeta_C \not\models t\approx t'$ for some free atom $t\approx t' \in \Gamma$. Hence, $t$ and $t'$ are either variables or constant symbols of the free sort, which means they do not contain subterms of the base sort. Since $\cB$ and $\cA$ behave identical on free-sort constant symbols and $\beta(u) = \hbeta_C(u)$ for every variable $u\in V_\cS$, we get $\cB, \beta \not\models t\approx t'$.

			\item Case $\cA, \hbeta_C \models t\approx t'$ for some $t\approx t' \in \Delta$. Analogous to the above case, $\cB, \beta \models t\approx t'$.
			
			\item Case $\cA, \hbeta_C \not\models P(t'_1, \ldots, t'_{m'}, t_1, \ldots, t_m)$ for some non-equational atom\\ $P(t'_1, \ldots, t'_{m'}, t_1, \ldots, t_m) \in \Gamma$. 
				This translates to \\
					\centerline{$\bigl\< \cA(\hbeta_C)(t'_1), \ldots, \cA(\hbeta_C)(t'_{m'}), \cA(\hbeta_C)(t_1), \ldots, \cA(\hbeta_C)(t_m)\bigr\> \not\in P^\cA$.}
				By construction of $\hbeta_C$, we have $\cA(\hbeta_C)(t_j) \in \hQ$ for every $j$, $1 \leq j \leq m$.
				Due to our assumptions regarding $\hQ$ and by construction of $\cB$, we therefore have \\
					\centerline{$\bigl\< \cA(\hbeta_C)(t'_1), \ldots, \cA(\hbeta_C)(t'_{m'}), \cA(\hbeta_C)(t_1), \ldots, \cA(\hbeta_C)(t_m) \bigr\> \not\in P^{\cB}$.}
				We observe the following properties:
				\begin{itemize}
					\item We have $\cA(\hbeta_C)(t'_j) = \cB(\beta)(t'_j)$ for every $j$, $1 \leq j \leq m'$, due to the definition of $\cB$ and $\hbeta_C$.
					\item Since all the $t_j$ are base-sort variables, we get $\cA(\hbeta_C)(t_j) = \cB(\hbeta_C)(t_j)$ for every $j$, $1 \leq j \leq m$.
				\end{itemize}				
				These two observations yield \\
					\centerline{$\bigl\< \cB(\beta)(t'_1), \ldots, \cB(\beta)(t'_{m'}), \cB(\hbeta_C)(t_1), \ldots, \cB(\hbeta_C)(t_m) \bigr\> \not\in P^{\cB}$.}
				Because of this result, and due to $\hsimeq_\kappa$-uniformity of $\cB$, \\
					\centerline{$\bigl\< \cB(\hbeta_C)(t_1), \ldots, \cB(\hbeta_C)(t_m) \bigr\> \hsimeq_\kappa \bigl\< \cB(\beta)(t_1), \ldots, \cB(\beta)(t_m) \bigr\>$}
				leads to \\
					\centerline{$\bigl\< \cB(\beta)(t'_1), \ldots, \cB(\beta)(t'_{m'}), \cB(\beta)(t_1), \ldots, \cB(\beta)(t_m) \bigr\> \not\in P^{\cB}$.}
					
				Put differently, we have $\cB, \beta \not\models P(t'_1, \ldots, t'_{m'}, t_1, \ldots, t_m)$.
													
			\item Case $\cA, \hbeta_C \models P(t'_1, \ldots, t'_{m'}, t_1, \ldots, t_m)$ for some non-equational atom\\
				$P(t'_1, \ldots, t'_{m'}, t_1, \ldots, t_m)\in\Delta$. 
				Analogously to the previous case we can infer $\cB, \beta \models P(t'_1, \ldots, t'_{m'}, t_1, \ldots, t_m)$.		
		\end{description}		
	Altogether, we have shown $\cB \models N$.	
\end{proof}

\subsubsection*{Proof of Lemma~\ref{lemma:ExistenceOfIndistinguishableFractionalParts}}

We first need the following auxiliary result.

\begin{lemma}\label{lemma:CorrespondenceSynchronous}
	Let $S \in (-\kappa-1, \kappa+1)^m/_{\simeq_\kappa}$ be an equivalence class with respect to $\simeq_\kappa$. There are two mappings $\varrho : [m] \to \{0, 1, \ldots, m\}$ and $\sigma : [m] \to \{-\kappa-1, \ldots, 0, \ldots, \kappa\}$ such that
	\begin{enumerate}[label=(\roman{*}), ref=(\roman{*})]
		\item\label{enum:lemmaCorrespondenceSynchronous:One} for any ascending tuple $\<r_0, r_1, \ldots, r_m\> \in [0,1)^{m+1}$ with $r_0 = 0$ we have $\bigl\<r_{\varrho(1)} + \sigma(1), \ldots, r_{\varrho(m)} + \sigma(m)\bigr\> \in S$, and
			
		\item\label{enum:lemmaCorrespondenceSynchronous:Two} for every tuple $\<s_1, \ldots, s_m\> \in S$ there is an ascending tuple $\<r_0, r_1, \ldots, r_m\> \in [0,1)^{m+1}$ with $r_0 = 0$ such that $\bigl\<s_1, \ldots, s_m\bigr\> = \bigl\<r_{\varrho(1)} + \sigma(1), \ldots, r_{\varrho(m)}+\sigma(m)\bigr\>$.
	\end{enumerate}
\end{lemma}
\begin{proof}
	Fix some tuple $\bq$ taken from $S$.
	Given $\bq$, we set $q'_0 := 0$ and further construct the sequence $q'_1, q'_2, \ldots$ in such a way that it lists all strictly positive fractional values in $\Fr(\bq)$ in ascending order.
	
	We construct $\sigma$ by setting $\sigma(i) := \lfloor q_i \rfloor$ for every $i = 1, \ldots, m$, 
	and $\varrho$ such that $\varrho(i) = k$ holds if and only if $\Fr(q_i) = q'_k$.
	Consequently, we have 
	\begin{itemize}
		\item[($*$)] $\<q_1, \ldots, q_m\> = \bigl\< \Fr(q_1) + \floor{q_1}, \ldots, \Fr(q_m) + \floor{q_m} \bigr\> = \bigl\< q'_{\varrho(1)} + \sigma(1), \ldots, q'_{\varrho(m)} + \sigma(m) \bigr\>$.
	\end{itemize}
	
	Let $\<r_0, r_1, \ldots, r_m\> \in [0,1)^{m+1}$ be any ascending tuple with $r_0 = 0$. 
	For all $i,j$, we observe the following properties:
		\begin{enumerate}[label=(\arabic{*}), ref=(\arabic{*})]
			\item $\floor{r_{\varrho(i)} + \sigma(i)} = \sigma(i) = \floor{q_i}$.
			
			\item $\Fr(r_{\varrho(i)} + \sigma(i)) = \Fr(r_{\varrho(i)}) = r_{\varrho(i)}$.
			
			\item $\varrho(i) = 0$ if and only if $\Fr(q_i) = q'_0 = 0$, which entails that $\Fr(r_{\varrho(i)} + \sigma(i)) = 0$ holds if and only if we have $\Fr(q_i) = 0$.

			\item $\Fr(q_i) = q'_{\varrho(i)}$.
									
			\item We have 
					$\Fr(r_{\varrho(i)} + \sigma(i)) \;\leq\; \Fr(r_{\varrho(j)} + \sigma(j))$ \\
				if and only if 
					$r_{\varrho(i)} \;\leq\; r_{\varrho(j)}$ \\
				if and only if 
					$\varrho(i) \;\leq\; \varrho(j)$ \\
				if and only if 
					$q'_{\varrho(i)} \;\leq\; q'_{\varrho(j)}$ \\
				if and only if 
					$\Fr(q_i) \;\leq\; \Fr(q_j)$.
		\end{enumerate}
		Taken together, these observations imply $\bq \simeq_\kappa \<r_{\varrho(1)} + \sigma(1), \ldots, r_{\varrho(m)} + \sigma(m)\>$. Hence, we have just proved \ref{enum:lemmaCorrespondenceSynchronous:One}.
		
		In fact, we have also already proved \ref{enum:lemmaCorrespondenceSynchronous:Two}, by giving the construction of the sequence $q'_0, q'_1, q'_2, \ldots$ and by having derived ($*$). If the sequence $q'_1, q'_2, \ldots$ is shorter than $m$ elements, we can simply pad it in an ascending fashion with arbitrary values from $(0,1)$.
\end{proof}

We can now prove Lemma~\ref{lemma:ExistenceOfIndistinguishableFractionalParts}.

\begin{lemma*}
	Let $\cA$ be an  interpretation and let $\kappa, \lambda$ be positive integers.
	There exists a finite set $Q \subset [0,1)$ of cardinality $\lambda+1$ such that $0 \in Q$ and for all tuples $\bs, \bs' \in \hQ^m$,
		$\bs \simeq_\kappa \bs'$ entails $\chi_\cA(\bs) = \chi_\cA(\bs')$,
	where \\
		\centerline{$\hQ := \bigl\{ q + k \bigm| q \in Q \text{ and } k \in \{-\kappa-1, \ldots, 0, \ldots, \kappa \} \bigr\}$.}
\end{lemma*}
\begin{proof}
	Let $S_1, \ldots, S_k$ be some enumeration of all equivalence classes in $(-\kappa-1, \kappa+1)^m/_{\simeq_\kappa}$.
	By Lemma~\ref{lemma:CorrespondenceSynchronous}, there is a (not necessarily unique) sequence $\<\varrho_1, \sigma_1\>, \ldots, \<\varrho_k, \sigma_k\>$ of pairs of functions such that each pair $\<\varrho_j, \sigma_j\>$ corresponds to the equivalence class $S_j$ in the sense of Lemma~\ref{lemma:CorrespondenceSynchronous}.
	
	Let $\hcS := \{ \fa \in \cS^\cA \mid \text{$\fa = c^\cA$ for some $c \in \fconsts(N)$}\}$ be the set of all domain elements assigned to free-sort constant symbols by $\cA$.
	We define a coloring $\hchi : \Real^m \to \bigl( \cP\{P_i \ba \mid \text{$\ba \in \hcS^{m'}$ and $P_i$ occurs}$ $\text{in $N$}\} \bigr)^k$ by setting 
		\begin{align*}
			\hchi(\br) := \bigl\< \chi_\cA &\bigl( \<r_{\varrho_1(1)} + \sigma_1(1), \ldots, r_{\varrho_1(m)} + \sigma_1(m)\> \bigr),\\
					 &\ldots, \chi_\cA \bigl( \<r_{\varrho_k(1)} + \sigma_k(1), \ldots, r_{\varrho_k(m)} + \sigma_k(m)\> \bigr) \bigr\>
		\end{align*}	
	for every tuple $\br = \<r_1, \ldots, r_m\> \in (0,1)^m$, where we define $r_0$ to be $0$.
	By virtue of Lemma~\ref{lemma:RamseyBasic:One}, there is a set $Q' \subseteq (0,1)$ of cardinality $\lambda$ such that all ascending tuples $\<r_1, \ldots, r_m\> \in {Q'}^m$ are assigned the same color by $\chi$. 
	We then set $Q := Q' \cup \{0\}$.
	
	Consider any equivalence class $S_j$ and the corresponding pair $\<\varrho_j, \sigma_j\>$ and let $\bs, \bs' \in \hQ^m$ be two $\simeq_\kappa$-equivalent tuples.
	Let $q_1, q_2, \ldots$ be an enumeration of all the strictly positive fractional parts in $\Fr(\bs)$ in ascending order and let $q_0 := 0$.	
	Hence, $q_0 < q_1 < q_2 < \ldots$.
	
	By Lemma~\ref{lemma:CorrespondenceSynchronous}, there are two ascending tuples $\bq := \<0, q_1, \ldots, q_m\>$ and $\bq' := \<0, q'_1, \ldots, q'_m\>$ in $[0,1)^{m+1}$ such that \\
		\centerline{$\bs = \<q_{\varrho_j(1)} + \sigma(1), \ldots, q_{\varrho_j(m)} + \sigma(m) \>$}
	and \\
		\centerline{$\bs' = \<q'_{\varrho_j(1)} + \sigma(1), \ldots, q'_{\varrho_j(m)} + \sigma(m) \>$.}
	Because of $\bs, \bs' \in \hQ^m$, we know that $\<q_1, \ldots, q_m\> \in {Q'}^m$ and $\<q'_1, \ldots, q'_m\> \in {Q'}^m$.
	Then, $\hchi(\<q_1, \ldots, q_m\>) = \hchi(\<q'_1, \ldots, q'_m\>)$ entails
		\begin{align*}
			\strut\hspace{10ex}
			\chi_\cA(\bs) &= \chi_\cA \bigl( \<q_{\varrho_j(1)} + \sigma(1), \ldots, q_{\varrho_j(m)} + \sigma(m) \> \bigr) \\
					 &= \chi_\cA \bigl( \<q'_{\varrho_j(1)} + \sigma(1), \ldots, q'_{\varrho_j(m)} + \sigma(m) \> \bigr) = \chi_\cA(\bs') ~.
			\qedhere
		\end{align*}	
\end{proof}

\subsection{Details Concerning Section~\ref{section:ReachabilityForTimedAutomata}}\label{section:appendix:TAReachability}
\subsubsection*{Proof of Lemma~\ref{lemma:DelayClausesWithDifferenceBounds}}

We first need an auxiliary result.

\begin{lemma}\label{lemma:DelaySetsViaDifferenceBounds}
	Let $S \in [0, \lambda+1)^{|\bx|}/_{\simeq_\lambda}$ be some equivalence class with respect to $\simeq_\lambda$.
	We define the two sets $\hS_1, \hS_2$ as follows:
		\begin{align*}
			\hS_1 := \bigl\{ \bq' \in [0, \lambda+1)^{|\bx|} \bigm|\, 
					&\text{there is some $\bq \in S$ such that for every $i$, $1 \leq i \leq |\bx|$,} \\
					&\text{we have $q_i \leq q'_i$ and $q'_0 - q'_i = q_0 - q_i$} \bigl\} ~,
		\end{align*}
	and	
		\begin{align*}
			\hS_2 := \bigl\{ \bq' \in [0&, \lambda+1)^{|\bx|} \bigm| \\
				  	&\text{there is some $\bq \in S$ such that for all $i_1, i_2$, $1 \leq i_1, i_2 \leq |\bx|$,} \\
					&\text{$q_{i_1} \leq q'_{i_1}$ and for every integer $k$, $-\lambda \leq k \leq \lambda$, we have} \\
					&\qquad\text{$q_{i_1} - q_{i_2} \leq k$ if and only if $q'_{i_1} - q'_{i_2} \leq k$, and}\\
					&\qquad\text{$q_{i_1} - q_{i_2} \geq k$ if and only if $q'_{i_1} - q'_{i_2} \geq k$} \bigl\} ~,
		\end{align*}
	where $q_0, q'_0$ are some fixed reals in the tuples $\bq, \bq'$, respectively, that correspond to the same index.	
	We observe $\hS_1 = \hS_2$.
\end{lemma}
\begin{proof}				
	We obviously have $\hS_1 \subseteq \hS_2$.
	
	In order to prove $\hS_2 \subseteq \hS_1$, consider any $\bq' \in \hS_2$.
	Pick some $\bs \in S$ for which $s_i \leq q'_i$ for every $i$, $1 \leq i \leq |\bx|$.
	By construction of $\hS_2$, we observe $\floor{s_0 - s_i} = \floor{q'_0 - q'_i}$ and $\ceil{s_0 - s_i} = \ceil{q'_0 - q'_i}$  for every $i$, $1 \leq i \leq |\bx|$.
	\begin{description}
			
		\item \underline{Claim III:}
			For all indices $j_1, j_2 \in \{1, \ldots, |\bx|\}$ we have $\Fr(s_{j_1}) = \Fr(s_{j_2})$ if and only if $\Fr(q'_{j_1}) = \Fr(q'_{j_2})$.
		\item \underline{Proof:}
			For all reals $r,t$ we have $\Fr(r) = \Fr(t)$ if and only if $\floor{r - t} = \ceil{r - t}$.
			Using this fact, we get that
				$\Fr(s_{j_1}) = \Fr(s_{j_2})$
			entails	
				$\floor{q'_{j_1} - q'_{j_2}} = \floor{s_{j_1} - s_{j_2}} = \ceil{s_{j_1} - s_{j_2}} = \ceil{q'_{j_1} - q'_{j_2}}$
			which in turn implies
				$\Fr(q'_{j_1}) = \Fr(q'_{j_2})$.
			Symmetrically, 	
				$\Fr(q'_{j_1}) = \Fr(q'_{j_2})$
			entails	
				$\Fr(s_{j_1}) = \Fr(s_{j_2})$.
			\strut\hfill$\Diamond$
		
		\item \underline{Claim IV:}
			Let $k_1, \ldots, k_{|\bx|}$be some enumeration of the indices in $\{1, \ldots, |\bx|\}$ such that $\Fr(s_{k_1}) \leq \ldots \leq \Fr(s_{k_{|\bx|}})$.
			There is some $\ell$ such that \\
			\centerline{$\Fr(q'_{k_{\ell+1}}) \leq \ldots \leq \Fr(q'_{k_{|\bx|}}) \leq \Fr(q'_{k_1}) \leq \ldots \leq \Fr(q'_{k_\ell})$.}
		\item \underline{Proof:}			
			Suppose  Claim~IV does not hold, while Claim~III is respected.
			Hence, suppose there are indices $j_1, j_2, j_3 \in \{1, \ldots, |\bx|\}$ such that 
				$\Fr(s_{j_1}) < \Fr(s_{j_2}) < \Fr(s_{j_3})$  
			and
				$\Fr(q'_{j_3}) < \Fr(q'_{j_2}) < \Fr(q'_{j_1})$.\footnote{There are analogous arguments leading to contradictions in the cases where $\Fr(q'_{j_2}) < \Fr(q'_{j_1}) < \Fr(q'_{j_3})$ and $\Fr(q'_{j_1}) < \Fr(q'_{j_3}) < \Fr(q'_{j_2})$.}

			For all reals $r, t$ we have $\floor{r-t} = \floor{r} - \floor{t} + \floor{\Fr(r) - \Fr(t)}$, where \\
				\centerline{$\floor{\Fr(r) - \Fr(t)} = \begin{cases} 0 &\text{if $\Fr(r) \geq \Fr(t)$} \\ -1 &\text{if $\Fr(r) < \Fr(t)$.} \end{cases}$}
			Hence, we get the following system of equations:
			\begin{center}
				$
				\begin{array}{r@{$\;\;=\;\;$}l@{$\qquad=\qquad$}r@{$\;\;=\;\;$}l}
					\floor{s_{j_1}} - \floor{s_{j_2}} - 1 & \floor{s_{j_1} - s_{j_2}} 	& 	\floor{q'_{j_1} - q'_{j_2}} & \floor{q'_{j_1}} - \floor{q'_{j_2}} \\ 
					\floor{s_{j_1}} - \floor{s_{j_3}} - 1 & \floor{s_{j_1} - s_{j_3}} 	& 	\floor{q'_{j_1} - q'_{j_3}} & \floor{q'_{j_1}} - \floor{q'_{j_3}} \\ 
					\floor{s_{j_2}} - \floor{s_{j_3}} - 1 & \floor{s_{j_2} - s_{j_3}} 	& 	\floor{q'_{j_2} - q'_{j_3}} & \floor{q'_{j_2}} - \floor{q'_{j_3}}
				\end{array}
				$
			\end{center}	
			As this system entails $0 = 1$, we obtain a contradiction.
			\strut\hfill$\Diamond$
	\end{description}

	It remains to prove the existence of some tuple $\bq$ that satisfies the following requirements:
	\begin{enumerate}[label=(\roman{*}), ref=(\roman{*})]
		\item\label{enum:proofAlternativeDelayClauses:II:I} $\floor{\bq} = \floor{\bs}$ and $\ceil{\bq} = \ceil{\bs}$.
		\item\label{enum:proofAlternativeDelayClauses:II:II} $\floor{s_0 - s_i} = \floor{q_0 - q_i}$ and $\ceil{s_0 - s_i} = \ceil{q_0 - q_i}$ for every $i$.
		\item\label{enum:proofAlternativeDelayClauses:II:III} $q_0 - q_i = q'_0 - q'_i$ for every $i$.
		\item\label{enum:proofAlternativeDelayClauses:II:IV} $q_i \leq q'_i$ for every $i$.
	\end{enumerate}
	Notice that Requirement~\ref{enum:proofAlternativeDelayClauses:II:II} is entailed by Requirement~\ref{enum:proofAlternativeDelayClauses:II:III} and the definition of $S_2$.
	
	Consider any $i$ with $1 \leq i \leq |\bx|$.
	Requirement~\ref{enum:proofAlternativeDelayClauses:II:I} entails that $\bq$ must satisfy $q_i = \floor{s_i} + \Fr(q_i)$.
	It follows that $q_0 - q_i = \floor{s_0} + \Fr(q_0) - \floor{s_i} - \Fr(q_i)$ and $q'_0 - q'_i = \floor{q'_0} + \Fr(q'_0) - \floor{q'_i} - \Fr(q'_i)$.
	Hence, Requirement~\ref{enum:proofAlternativeDelayClauses:II:III} entails \\
		\centerline{$\floor{s_0} - \floor{s_i} + \Fr(q_0) - \Fr(q_i) \;\;=\;\;\ \floor{q'_0} - \floor{q'_i} + \Fr(q'_0) - \Fr(q'_i)$,}
	which is equivalent to
		\begin{equation*}(*)\qquad
			\Fr(q_0) - \Fr(q_i) \;\;=\;\;(\floor{q'_0} - \floor{q'_i}) - (\floor{s_0} - \floor{s_i}) + \Fr(q'_0) - \Fr(q'_i) ~.
		\end{equation*}

	We distinguish several cases:
	\begin{description}
		\item If $\bq' \in S$, then we set $\bq := \bq'$.
		\item If there is some $j$ such that $\floor{s_j} = \ceil{s_j}$, then, by Requirement~\ref{enum:proofAlternativeDelayClauses:II:I}, we must satisfy $\Fr(q_j) = 0$ and, therefore, for every $i$, $\Fr(q_i)$ is determined by $(*)$.
		\item If $\Fr(s_1) = \ldots = \Fr(s_{|\bx|})$, then we observe $\floor{q'_0 - q'_i} = \floor{s_0 - s_i} = \ceil{s_0 - s_i} = \ceil{q'_0 - q'_i}$ for every $i$.
			Hence, we have $\floor{q'_0 - q'_i} = \ceil{q'_0 - q'_i}$, which implies $\Fr(q'_0) = \Fr(q'_i)$ for every $i$.
			As this entails $q'_0 - q'_i = \floor{q'_0 - q'_i} = \floor{s_0 - s_i} = s_0 - s_i$, Requirement~\ref{enum:proofAlternativeDelayClauses:II:III} is satisfied if we set $\bq := \bs$.

		\item If none of the above cases apply, we have $\floor{s_i} = \ceil{s_i} - 1$ for every $i$.
			Moreover, we know that there are indices $i_1, i_2$ such that $\Fr(s_{i_1}) < \Fr(s_{i_2})$.
									
			Let $k_1, \ldots, k_{|\bx|}$be some enumeration of the indices in $\{1, \ldots, |\bx|\}$ such that $\Fr(s_{k_1}) \leq \ldots \leq \Fr(s_{k_{|\bx|}})$.
			Notice that $\Fr(s_{k_1}) < \Fr(s_{k_{|\bx|}})$ holds due to our assumptions.
			By Claim~IV, there is some $\ell$ such that \\
			\centerline{$\Fr(q'_{k_{\ell+1}}) \leq \ldots \leq \Fr(q'_{k_{|\bx|}}) \leq \Fr(q'_{k_1}) \leq \ldots \leq \Fr(q'_{k_\ell})$.}
			
			In fact, Claim~III together with $\Fr(s_{k_1}) < \Fr(s_{k_{|\bx|}})$ entails that $\Fr(q'_{k_{|\bx|}})$ is strictly smaller than $\Fr(q'_{k_1})$.

			We pick some real $\varepsilon$ such that $0 < \varepsilon < \Fr(q'_{k_1}) - \Fr(q'_{k_{|\bx|}})$.
			For every $j$, $1 \leq j \leq \ell$, we set 
				$\Fr(q_{k_j}) := \varepsilon + \bigl( \Fr(q'_{k_j}) - \Fr(q'_{k_1}) \bigr)$. 
			For every $j$, $\ell+1 \leq j \leq |\bx|$, we set 
				$\Fr(q_{k_j}) := \varepsilon + 1 - \bigl( \Fr(q'_{k_1}) - \Fr(q'_{k_j}) \bigr)$. 
			\begin{description}
				\item \underline{Claim V:}
					$0 < \Fr(q_{k_1}) \leq \ldots \leq \Fr(q_{k_\ell}) \leq \Fr(q_{k_{\ell+1}}) \leq \ldots \leq \Fr(q_{k_{|\bx|}}) < 1$.
				\item \underline{Proof:}

					We observe 
						\begin{itemize}
							\item $\Fr(q_{k_1}) = \varepsilon+\bigl( \Fr(q'_{k_1}) - \Fr(q'_{k_1}) \bigr) = \varepsilon > 0$.
							\item $\Fr(q_{k_{|\bx|}}) = \varepsilon + 1 - \bigl( \Fr(q'_{k_1}) - \Fr(q'_{k_{|\bx|}}) \bigr) < \bigl( \Fr(q'_{k_1}) - \Fr(q'_{k_{|\bx|}}) \bigr) + 1 - \bigl( \Fr(q'_{k_1}) - \Fr(q'_{k_{|\bx|}}) \bigr) = 1$.
							\item Because of $\Fr(q'_{k_\ell}) \in [0,1)$ and $\Fr(q'_{k_{\ell+1}}) \in [0,1)$, we obtain $\Fr(q'_{k_\ell}) \leq \Fr(q'_{k_{\ell+1}}) + 1$.
								Hence, we get
								$	\Fr(q_{k_\ell}) 
									= \varepsilon + \bigl( \Fr(q'_{k_\ell}) - \Fr(q'_{k_1}) \bigr) 
									\leq \varepsilon + \Fr(q'_{k_{\ell+1}}) + 1 - \Fr(q'_{k_1}) 
									= \varepsilon + 1 - \bigl( \Fr(q'_{k_1}) - \Fr(q'_{k_{\ell+1}}) \bigr) 
									= \Fr(q_{k_{\ell+1}})
								$.
						\end{itemize}
						The above observations entail $0 < \Fr(q_{k_1})$, $\Fr(q_{k_\ell}) \leq \Fr(q_{k_{\ell+1}})$, and $\Fr(q_{k_{|\bx|}}) < 1$.
						By definition of the $\Fr(q_{k_j})$ and our assumptions $\Fr(q'_{k_1}) \leq \ldots \leq \Fr(q'_{k_\ell})$ and $\Fr(q'_{k_{\ell+1}}) \leq \ldots \leq \Fr(q'_{k_{|\bx|}})$, these observations imply Claim~V. 
					\strut\hfill$\Diamond$
						
				\item \underline{Claim VI:}
					For every $j$ we have
						$\bigl( \floor{s_{k_1}} + \Fr(q_{k_1}) \bigr) - \bigl( \floor{s_{k_j}} + \Fr(q_{k_j}) \bigr) = q'_{k_1} - q'_{k_j}$.
				\item \underline{Proof:}
					If $1 \leq j \leq \ell$, then we have 
						\begin{align*}
							&\bigl( \floor{s_{k_1}} + \Fr(q_{k_1}) \bigr) - \bigl( \floor{s_{k_j}} + \Fr(q_{k_j}) \bigr) \\
							&= \floor{s_{k_1}} + \varepsilon + \bigl( \Fr(q'_{k_1}) - \Fr(q'_{k_1}) \bigr) - \floor{s_{k_j}} - \varepsilon - \bigl( \Fr(q'_{k_j}) - \Fr(q'_{k_1}) \bigr) \\
							&= \floor{s_{k_1}} - \floor{s_{k_j}} + \Fr(q'_{k_1}) - \Fr(q'_{k_j}) \\
							&= \floor{s_{k_1} - s_{k_j}} + \delta + \Fr(q'_{k_1}) - \Fr(q'_{k_j}) \\
							&= \floor{q'_{k_1} - q'_{k_j}} + \delta + \Fr(q'_{k_1}) - \Fr(q'_{k_j}) \\
							&= \floor{q'_{k_1}} -\floor{q'_{k_j}} + \Fr(q'_{k_1}) - \Fr(q'_{k_j}) \\
							&= q'_{k_1} - q'_{k_j} ~,
						\end{align*}	\
					where $\delta = 0$ in case of $\Fr(q'_{k_1}) = \Fr(q'_{k_j})$ (or $\Fr(s_{k_1}) = \Fr(s_{k_j})$) and $\delta = 1$ if $\Fr(q'_{k_1}) < \Fr(q'_{k_j})$ (or $\Fr(s_{k_1}) < \Fr(s_{k_j})$).
					
					If $\ell + 1 \leq j \leq |\bx|$, then we have 
						\begin{align*}
							&\bigl( \floor{s_{k_1}} + \Fr(q_{k_1}) \bigr) - \bigl( \floor{s_{k_j}} + \Fr(q_{k_j}) \bigr) \\
							&= \floor{s_{k_1}} + \varepsilon + \bigl( \Fr(q'_{k_1}) - \Fr(q'_{k_1}) \bigr) - \floor{s_{k_j}} - \varepsilon - 1 + \bigl( \Fr(q'_{k_1}) - \Fr(q'_{k_j}) \bigr) \\
							&= \floor{s_{k_1}} - \floor{s_{k_j}} + \Fr(q'_{k_1}) - \Fr(q'_{k_j}) -1 ~.
						\end{align*}
					Since $\Fr(q'_{k_{|\bx|}})$ is strictly smaller than $\Fr(q'_{k_1})$, we get $\Fr(q'_{k_j}) < \Fr(q_{k_1})$.
					Moreover, Claim~III together with $\Fr(s_{k_1}) \leq \Fr(s_{k_j})$ entails $\Fr(s_{k_1}) < \Fr(s_{k_j})$.
					Hence, $\floor{s_{k_1}} - \floor{s_{k_j}} = \floor{s_{k_1} - s_{k_j}} + 1$ and $\floor{q'_{k_1} - q'_{k_j}} = \floor{q'_{k_1}} - \floor{q'_{k_j}}$.
					Consequently, we get
						\begin{align*}
							&\floor{s_{k_1}} - \floor{s_{k_j}} + \Fr(q'_{k_1}) - \Fr(q'_{k_j}) -1 \\
							&= \floor{s_{k_1} - s_{k_j}} + 1 + \Fr(q'_{k_1}) - \Fr(q'_{k_j}) - 1 \\
							&= \floor{q'_{k_1} - q'_{k_j}} + \Fr(q'_{k_1}) - \Fr(q'_{k_j}) \\
							&= \floor{q'_{k_1}} -\floor{q'_{k_j}} + \Fr(q'_{k_1}) - \Fr(q'_{k_j}) \\
							&= q'_{k_1} - q'_{k_j}.
						\end{align*}					
					\strut\hfill$\Diamond$

				\item \underline{Claim VII:}
					For every $j$ we have $\floor{s_{k_j}} + \Fr(q_{k_j}) \leq q'_{k_j}$.
						
				\item \underline{Proof:}
					As we assume $\bq' \not\in S$, there is at least one $i$ such that $\floor{q'_{k_i}} > \floor{s_{k_i}}$.
					This entails $q'_{k_i} \geq \floor{q'_{k_i}} > \floor{s_{k_i}} + \Fr(q_{k_i})$.
					As one consequence of Claim~VI, we get that $\bigl( \floor{s_{k_i}} + \Fr(q_{k_i}) \bigr) - \bigl( \floor{s_{k_j}} + \Fr(q_{k_j}) \bigr) = q'_{k_i} - q'_{k_j}$ for every $j$.
					This can be rewritten into the equivalent equation $q'_{k_j} - \bigl( \floor{s_{k_j}} + \Fr(q_{k_j}) \bigr) = q'_{k_i} - \bigl( \floor{s_{k_i}} + \Fr(q_{k_i}) \bigr)$.
					In other words, we have $q'_{k_j} > \bigl( \floor{s_{k_j}} + \Fr(q_{k_j}) \bigr)$ for every $j$.
					\strut\hfill$\Diamond$
			\end{description}
			This means, if we set $q_{k_j} := \floor{s_{k_j}} + \Fr(q_{k_j})$ for every $j$, $1 \leq j \leq |\bx|$, then the stipulated requirements are satisfied.	
			\qedhere
	\end{description}	
\end{proof}

Using the above result, we can prove Lemma~\ref{lemma:DelayClausesWithDifferenceBounds}
\begin{lemma*}
	Consider any delay clause 
		\[ 	C :=\quad 
			z\geq 0 \;\;\wedge\;\; 
			\bigwedge_{x \in \bx} x' \equals x + z \;\;\wedge\;\; 
			\inv_\ell[\bx'] 
			\;\bigm\|\; \Reach(\ell, \bx) \rightarrow \Reach(\ell, \bx') .
		\]
	that belongs to the FOL(LA) encoding of some timed automaton $\bA := \< \Loc, \ell_0, \bx,$ $\<\inv_\ell\>_{\ell\in \Loc}, \cT \>$.
	Let $\lambda$ be some positive integer.
	Let $M$ be a finite clause set corresponding to the following formula
		\begin{align*}
			\varphi :=\;\; \bigwedge_{x_1, x_2 \in \bx} &\bigwedge_{-\lambda \leq k \leq \lambda} \bigl(x_1 - x_2 \leq k \;\leftrightarrow\; x'_1 - x'_2 \leq k \bigr) \\[-0.75ex]
				&\hspace{13ex} \wedge\; \bigl(x_1 - x_2 \geq k \leftrightarrow x'_1 - x'_2 \geq k \bigr) \\[0.25ex]
				&\wedge\; 
					\bigwedge_{x \in \bx} x' \geq x 
					\;\;\wedge\;\; 
					\inv_\ell[\bx'] 
					\;\;\bigm\|\;\; \Reach(\ell, \bx) \rightarrow \Reach(\ell, \bx') .
		\end{align*}
	For every $\simeq_\lambda$-uniform interpretation $\cA$ we have 
		$\cA, [\bx \Mapsto \br, \bx' \Mapsto \br'] \models C$
		for all tuples $\br, \br' \in [0, \lambda+1)^{|\bx|}$
	if and only if
		$\cA, [\bx \Mapsto \bq, \bx' \Mapsto \bq'] \models M$
		holds for all tuples $\bq, \bq' \in [0, \lambda+1)^{|\bx|}$.
\end{lemma*}
\begin{proof}~
	We first show that the clause $C$ is equivalent to the clause
			\begin{align*}
				C' :=\;
				\bigwedge_{x \in \bx} x_0 - x = x'_0 - x' \;\;\wedge\; &\bigwedge_{x \in \bx} x' - x \geq 0 \\
				&\;\;\wedge\;\; 
				\inv_\ell[\bx'] 
				\;\bigm\|\; \Reach(\ell, \bx) \rightarrow \Reach(\ell, \bx') ~,
			\end{align*}
	where $x_0$ is some fixed clock variable $x_0 \in \bx$.		
	Although the variable $z$ in $C$ is universally quantified, the fact that $z$ does not occur on the right-hand side of the implication entails that $z$'s quantifier can be moved inside the premise of the implication $C$ represents, where universal quantification will turn into existential quantification (the quantifier moves into the scope of one implicit negation symbols).			
	This transformation yields an equivalent clause with the constraint $\exists z.\; \bigwedge_{x \in \bx} x' - x = z \;\wedge\; z \geq 0 \;\wedge\; 
	\inv_\ell[\bx']$.
	In addition, we observe
	\begin{align*}
		&\exists z.\; \bigwedge_{x \in \bx} x' - x = z \;\;\wedge\;\; z \geq 0 \\
		&\semequiv\;\; \bigwedge_{x_1,x_2 \in \bx} x'_1 - x_1 = x'_2 - x_2 \;\;\wedge\;\; \bigwedge_{x \in \bx} x' - x \geq 0 \\
		&\semequiv\;\; \bigwedge_{x \in \bx} x'_0 - x_0 = x' - x \;\;\wedge\;\; \bigwedge_{x \in \bx} x' - x \geq 0 \\
		&\semequiv\;\; \bigwedge_{x \in \bx} x_0 - x = x'_0 - x' \;\;\wedge\;\; \bigwedge_{x \in \bx} x' - x \geq 0 ~.
	\end{align*}
	Consequently, the clauses $C$ and $C'$ are equivalent.

	Let $S \subseteq [0, \lambda+1)^{|\bx|}$ be any equivalence class with respect to $\simeq_\lambda$.
	Since we assume $\cA$ to be $\simeq_\lambda$-uniform, we have that, if $\cA, [\bx \Mapsto \br] \models \Reach(\ell, \bx)$ holds for one $\br \in S$, then $\cA, [\bx \Mapsto \bq] \models \Reach(\ell, \bx)$ holds for every $\bq \in S$.

	Now suppose $\cA,[\bx \Mapsto \br, \bx' \Mapsto \br'] \models C'$ holds for all tuples $\br, \br' \in [0, \lambda+1)^{|\bx|}$.
	Moreover, suppose there is some pair of tuples $\bq, \bq' \in [0, \lambda+1)^{|\bx|}$ such that $\cA,[\bx \Mapsto \bq, \bx' \Mapsto \bq'] \not\models \varphi$.
	Thus, we have that $\cA,[\bx \Mapsto \bq, \bx' \Mapsto \bq']$ satisfies the premises of $\varphi$---among them $\inv_\ell[\bx']$--- but does not satisfy the consequent $\Reach(\ell, \bx')$.
	As $\cA, [\bx \Mapsto \bq, \bx' \Mapsto \bq'] \models 
		\bigwedge_{x_1, x_2 \in \bx} \bigwedge_{-\lambda \leq k \leq \lambda} \bigl(x_1 - x_2 \leq k \;\leftrightarrow\; x'_1 - x'_2 \leq k \bigr)
		\;\wedge\; \bigl(x_1 - x_2 \geq k \leftrightarrow x'_1 - x'_2 \geq k \bigr)
		\;\wedge\; \bigwedge_{x \in \bx} x' \geq x$,
	we know that $\bq' \in \hS_2$, where $S \subseteq [0, \lambda+1)^{|\bx|}$ is the equivalence class with respect to $\simeq_\lambda$ to which $\bq$ belongs and $\hS_2$ is defined as in Lemma~\ref{lemma:DelaySetsViaDifferenceBounds}.
	Moreover, we know that $\cA, [\bx \Mapsto \bs] \models \Reach(\ell, \bx)$ for every $\bs \in S$, as $\cA$ is $\simeq_\lambda$-uniform.
	The fact that $\cA, [\bx \Mapsto \br, \bx' \Mapsto \br'] \models C'$ holds for all tuples $\br, \br' \in [0, \lambda+1)^{|\bx|}$ entails $\cA, [\bx' \Mapsto \bs'] \models \Reach(\ell, \bx')$ for every $\bs' \in \hS_1$ for which $[\bx' \Mapsto \bs'] \models \inv_\ell[\bx']$, where $\hS_1$ is defined as in Lemma~\ref{lemma:DelaySetsViaDifferenceBounds}.
	Hence, Lemma~\ref{lemma:DelaySetsViaDifferenceBounds} entails $\cA, [\bx' \Mapsto \bs''] \models \Reach(\ell, \bx')$ for every $\bs'' \in \hS_2$ for which $[\bx' \Mapsto \bs''] \models \inv_\ell[\bx']$, in particular for $\bs'' = \bq'$.
	This contradiction implies that $\cA, [\bx \Mapsto \bq, \bx' \Mapsto \bq'] \models \varphi$ holds for all tuples $\bq, \bq' \in [0, \lambda+1)^{|\bx|}$.
	
	The opposite direction can be argued analogously.
\end{proof}

\subsubsection*{Details regarding Theorem~\ref{theorem:TAReachabilityInBdrBd}}

Figure~\ref{figure:EquivalenceClassesForTAs} illustrates the TA regions for some timed automaton with two clocks and in which all integer constants have an absolute value of at most $2$.
For every TA region $R \subseteq \Real^2$ of such an automaton, there is at least one representative $\br \in R$ which lies in $[0, 5)^2$.
\begin{figure}[tbh]
\centerline{
\begin{tabular}{p{20pt}p{100pt}}
	\begin{picture}(20,0)
		\put(0,-5){\mbox{$\<0,0\>$}}
	\end{picture}
	&
	\begin{picture}(100, 115)
		\multiput(0,0)(20,0){3}{\circle*{3}}
		\multiput(0,20)(20,0){4}{\circle*{3}}
		\multiput(0,40)(20,0){4}{\circle*{3}}
		\multiput(20,60)(20,0){2}{\circle*{3}}
		\put(40,80){\circle*{3}}
		\put(80,40){\circle*{3}}
		\linethickness{1.5pt}
		\multiput(40,62)(20,0){1}{\line(0,1){16}}
		\multiput(20,42)(20,0){2}{\line(0,1){16}}
		\multiput(0,22)(20,0){3}{\line(0,1){16}}
		\multiput(0,2)(20,0){3}{\line(0,1){16}}
		\put(0,42){\line(0,1){76}}
		\put(20,62){\line(0,1){56}}
		\put(40,82){\line(0,1){36}}
		\multiput(2,40)(20,0){4}{\line(1,0){16}}
		\multiput(2,20)(20,0){4}{\line(1,0){16}}
		\multiput(2,0)(20,0){2}{\line(1,0){16}}
		\put(42,0){\line(1,0){76}}
		\put(62,20){\line(1,0){56}}
		\put(82,40){\line(1,0){36}}
		\multiput(2,42.2)(20,20){2}{\line(1,1){15.8}}
		\multiput(2,42)(20,20){2}{\line(1,1){16}}
		\multiput(2.2,42)(20,20){2}{\line(1,1){15.8}}
		\put(42,82.2){\line(1,1){35.8}}
		\put(42,82){\line(1,1){36}}
		\put(42.2,82){\line(1,1){35.8}}
		\multiput(2,22.2)(20,20){2}{\line(1,1){15.8}}
		\multiput(2,22)(20,20){2}{\line(1,1){16}}
		\multiput(2.2,22)(20,20){2}{\line(1,1){15.8}}
		\put(42,62.2){\line(1,1){55.8}}
		\put(42,62){\line(1,1){56}}
		\put(42.2,62){\line(1,1){55.8}}
		\multiput(2,2.2)(20,20){2}{\line(1,1){15.8}}
		\multiput(2,2)(20,20){2}{\line(1,1){16}}
		\multiput(2.2,2)(20,20){2}{\line(1,1){15.8}}
		\put(42,42.2){\line(1,1){75.8}}
		\put(42,42){\line(1,1){76}}
		\put(42.2,42){\line(1,1){75.8}}
		\multiput(22,2.2)(20,20){2}{\line(1,1){15.8}}
		\multiput(22,2)(20,20){2}{\line(1,1){16}}
		\multiput(22.2,2)(20,20){2}{\line(1,1){15.8}}
		\put(62,42.2){\line(1,1){55.8}}
		\put(62,42){\line(1,1){56}}
		\put(62.2,42){\line(1,1){55.8}}
		\multiput(42,2.2)(20,20){2}{\line(1,1){15.8}}
		\multiput(42,2)(20,20){2}{\line(1,1){16}}
		\multiput(42.2,2)(20,20){2}{\line(1,1){15.8}}
		\put(82,42.2){\line(1,1){35.8}}
		\put(82,42){\line(1,1){36}}
		\put(82.2,42){\line(1,1){35.8}}
	\end{picture}
\end{tabular}	
}
\caption{
	Partition of the set $[0,\infty)^2$ into classes of clock valuations that cannot be distinguished by a timed automaton with two clocks in which the absolute value of integer constants occurring in location invariants and transition guards does not exceed~$2$.
	Every dot, line segment, and white area represents an equivalence class.}
\label{figure:EquivalenceClassesForTAs}
\end{figure}
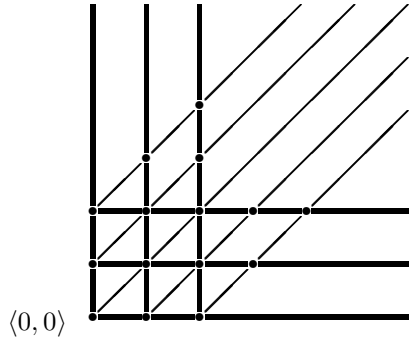

Let $\bA := \< \Loc, \ell_0, \bx,$ $\<\inv_\ell\>_{\ell\in \Loc}, \cT \>$ be a timed automaton and let $k$ be the maximal absolute value of any integer constant occurring in the invariants and the transition guards of $\bA$.
Let $x_1, \ldots, x_\ell$ be some enumeration of all the clock variables in $\bx$.
Consider a constraint of the form \\
	\centerline{$\psi :=\;\; x_1 - x_2 = k \;\wedge\; x_2 - x_3 = k \;\wedge \ldots \wedge\; x_{\ell-1} - x_\ell = k$.}
We observe that $\psi$ entails $x_1 - x_\ell = (\ell-1) \cdot k$.
Of course, $\psi$ can also be conjoined with the constraint $x_1 < -k$, say, which entails $x_\ell < -k-(\ell-1)\cdot k$.
This example illustrates that one can combine several difference constraints $x - y \rel c$ over different clock variables in such a way that bounds are achieved which cannot be formulated with a single constraint $u - v \rel d$ with $|d| \leq k$.
However, all of those combined constraints can be equivalently represented with atomic constraints $x - y \rel c$ or $x \rel c$, where $|c| \leq |\bx|\cdot k$.

In the main text (in the discussion preceding Theorem~\ref{theorem:TAReachabilityInBdrBd} in Section~\ref{section:ReachabilityForTimedAutomata}), we mention that there exists a computable integer $\lambda$ such that any valuation $\br$ of $\bA$'s clocks can be projected to some valuation $\br' \in [0, \lambda+1)^{|\bx|}$ which $\bA$ cannot distinguish from $\br$.
Due to the above observations, we find that $\lambda = |\bx|\cdot k$ meets the stipulated requirements.
Hence, in order to decide reachability for $\bA$, it is sufficient to consider the bounded subspace $[0,\lambda+1)^{|\bx|} \subseteq \Real^{|\bx|}$.
This means, given the FOL(LA) encoding $N_\bA$ of $\bA$, we obtain a \BsrBd\ encoding $N'_\bA$ of reachability with respect to $\bA$ in the following two steps: \\
(1) Replace every delay clause in $N_\bA$ with a corresponding finite set of clauses $M$ in accordance with Lemma~\ref{lemma:DelayClausesWithDifferenceBounds}, where we set $\lambda := |\bx|\cdot k$.\\
(2) Conjoin the constraints $0 \leq x \;\wedge\; x < \kappa$ for $\kappa := \lambda+1 = |\bx|\cdot k + 1$ to every constraint in which a base-sort variable $x$ occurs.\\
Since any $\hsimeq_{\lambda+1}$-uniform model of $N'_\bA$ is $\simeq_{\lambda}$-uniform over the subspace $(-\lambda-1, \lambda+1)^{|\bx|}$, Lemma~\ref{lemma:DelayClausesWithDifferenceBounds} entails that $N'_\bA$ faithfully encodes reachability with respect to $\bA$.

\end{document}